\documentclass[10pt,twocolumn,twoside]{IEEEtran}

\usepackage{amsmath}
\usepackage{tabularx}
\usepackage{graphicx}
\usepackage{subfigure}
\usepackage{amssymb}
\usepackage{amsthm}
\usepackage{epstopdf}
\usepackage{color}
\usepackage{algorithm,algpseudocode}
\newtheorem{theorem}{Theorem}
\newtheorem{corollary}{Corollary}

\newtheorem{lemma}{Lemma}
\newtheorem{definition}{Definition}

\usepackage{url} 
\usepackage{cite}

\usepackage{fancyhdr}
\usepackage{lipsum}

\begin{document}

\title{Control of Multi-Layer Mobile Autonomous Systems in Adversarial Environments: A Games-in-Games Approach}

\author{
          Juntao~Chen,~\IEEEmembership{Student Member,~IEEE,}
       and Quanyan~Zhu,~\IEEEmembership{Member,~IEEE}

\thanks{Juntao Chen and Quanyan Zhu are with the Department of Electrical and Computer Engineering, Tandon School of Engineering, New York University, Brooklyn, NY, 11201 USA. E-mail: \{jc6412, qz494\}@nyu.edu.}
\thanks{This paper has been accepted for publication in \textit{IEEE Transactions on Control of Network Systems}.}}

\maketitle

\begin{abstract}
Mobile autonomous system (MAS) becomes pervasive especially in the vehicular and robotic networks. Multiple heterogeneous MAS networks can be integrated together as a multi-layer MAS network to offer holistic services. The network connectivity of multi-layer MAS plays an important role in the information exchange between agents within and across different layers of the network. In this paper, we establish a \textit{games-in-games} framework to capture the uncoordinated nature of decision makings under adversarial environment at different layers. Specifically, each network operator controls the mobile agents in his own subnetwork and designs a secure strategy to maximize the global network connectivity by considering the behavior of jamming attackers that aim to disconnect the network. The solution concept of \textit{meta-equilibrium} is proposed to characterize the \textit{system-of-systems} behavior of the autonomous agents. For online implementation of the control, we design a resilient algorithm that improves the network algebraic connectivity iteratively. We show that the designed algorithm converges to a meta-equilibrium asymptotically. Finally, we use case studies of a two-layer MAS network to corroborate the security and agile resilience of the network controlled by the proposed strategy.
\end{abstract}



\section{Introduction}
Cooperative mobile autonomous system (MAS) has a wide range of applications, such as rescue and monitoring the crowd in mission critical scenarios \cite{erman2008enabling,erdelj2017help}. One of the challenges in designing the MAS network is to maintain the connectivity between agents/robots\footnote{The ``agent" refers to the robot in our MAS network. We also use the terms ``MAS network" and ``robotic network" interchangeably.}, since a higher connectivity enables faster information spreading and hence a high-level of situational awareness. Connectivity control of the mobile robotic networks has been addressed in a number of previous works including \cite{michael2009maintaining,kim2006maximizing,simonetto2013constrained}. They have successfully tackled a single network of cooperative robots. Recent advances in networked systems have witnessed emerging applications involving multi-layer networks or \textit{network-of-networks} \cite{d2014networks,martin2014algebraic,
pawlick2018istrict}. For example, in the battlefield, unmanned aerial vehicles (UAVs) and unmanned ground vehicles (UGVs) execute tasks together, and the whole network can be seen as a two-layer interdependent network. { The connectivity of the two-layer network can play a key role in the operations for a given collaborative mission. To enable the real-time decision making of each agent, the integrated network needs to guarantee a level of connectivity.} Therefore, the current single network control paradigm is not yet sufficient to address the challenges related to the analysis and design of multi-layer MAS networks. The main objective of this work is to develop a theoretic framework that can capture the interactions between agents within a network and across networks and enable the design of distributed control algorithms that can maintain the connectivity in both adversarial and non-adversarial environment.

In our problem, the operator of each layer MAS network aims to maximize the algebraic connectivity \cite{fiedler1973algebraic} of the global network. If the whole network is fully cooperative or governed by a single network operator, then the designed network is a \textit{team-optimal} solution. However, in practice, different layers of robotic networks are often operated by different entities, which makes the coordination between separate entities difficult. This uncoordinated control design naturally leads to a \textit{system-of-systems} (SoS) framework of the multi-layer MAS network. For example, in the aforementioned two-layer UAV and UGV mobile networks, though the objectives of two network operators are aligned, the UAVs are operated by one entity while the UGVs is operated by another. The lack of the centralized planning can result in insufficient coordination between two networks and lead to disruptions in connectivity and security vulnerabilities. To address this problem, we establish a \textit{Nash} game-theoretic model in which two players, i.e., network operators, control robots at their layer, to maximize the global connectivity independently. This model captures the lack of coordination between players and their decentralized decision making in optimizing the SoS performance.

Cybersecurity is another critical concern of MAS, since robots are prone to adversarial attacks, e.g., the communication links between robots can be jammed (e.g., \cite{xu2006jamming,chen2019optimal}) which decreases the connectivity. Therefore, secure control of the multi-layer MAS networks is critical to maintain the SoS performance at a high level. To this end, we model the mobility of robots by taking into account the imperfect communication links under adversarial environment. Specifically, each network operator anticipates the jamming attacks and controls robots by anticipating that a set of critical links between agents can be compromised. This secure control design can be modeled by a \textit{Stackelberg} game between each network operator and the attacker. 

In this work, we integrate the modeled Nash game between two network operators as well as the Stackelberg game between the network operator and the attacker which further yields a \textit{games-in-games} framework. This new type of game provides a holistic modeling that integrates the network-network interactions and the agent-adversary interactions together for the secure and decentralized control design of multi-layer MAS networks. We propose a \textit{meta-equilibrium} solution for this games-in-games which includes the optimal strategy of network operators and the strategic jamming attacks of adversary.
We further develop a resilient and decentralized mechanism that guides the algorithmic design of MAS for achieving a meta-equilibrium solution. A typical example is that two network operators update their own network in a round-robin fashion based on the current network to maximize the network connectivity. This alternating-play mechanism induces an iterative algorithm that can converge to a meta-equilibrium asymptotically. 

{The games-in-games framework provides a theoretical foundation for understanding the agile resilience of the system to cyberattacks, which is a critical system property for the MAS network to recover quickly especially for mission-critical applications. When robots or communication links in the network are compromised, the integrated MAS network under the designed control strategy needs to respond to the unexpected disruptions with agility to mitigate the loss of connectivity. Hence, to investigate the resilience of the designed algorithm,} we first introduce another type of cyber threats called GPS spoofing attacks. The resilience is measured by the enhanced SoS performance through the design of post-attack control strategies. Simulation results show that the multi-layer MAS network is resilient to attacks using the proposed control method. After the detection of spoofing attack by the network operator, the MAS network shows agile resilience to the attack and the system adapts and reconfigures itself to an efficient meta-equilibrium that coincides with the one without attack.  

The contributions of this paper are summarized as follows:
\begin{enumerate}
\item We establish a games-in-games framework that enables uncoordinated and secure control of MAS in the multi-layer networks under cyberattacks.
\item We propose a meta-equilibrium solution concept to capture the interdependent decision makings of network players in a holistic fashion, and characterize the strategic behavior of the attacker explicitly.
\item We design a resilient and decentralized iterative algorithm that aims to maximize the algebraic connectivity of the global MAS network, and show its asymptotic convergence to a meta-equilibrium.
\item We corroborate the security and resiliency of the designed control algorithm for multi-layer MAS network extensively by case studies.
\end{enumerate}

\subsection{Related Work}
MAS has been applied to a number of emerging fields. One of them is drone delivery \cite{dorling2017vehicle} which has attracted significant attention from industries.
Another one is the unmanned aerial vehicles (UAVs) assisted sensing and communication networks for disaster response and recovery \cite{erdelj2017help,erdelj2017wireless}. Since we focus on controlling groups of MAS, our work is also related to the classical control of multi-agent systems \cite{martin2014multiagent,olfati2007consensus,
nedic2010constrained,ni2010leader}. 

One critical factor needs to be considered for MAS network is its connectivity. When MAS is adopted in mission critical scenarios, such as battlefields and disaster-affected areas, a higher network connectivity provides a higher level of situational awareness \cite{chen2019dynamic}. In this work, we adopt Fiedler value \cite{fiedler1973algebraic}  (or called algebraic connectivity) as a metric to measure the network connectivity.
Maximizing the algebraic connectivity of networks has been investigated extensively in literature, including single-layer static network \cite{ghosh2006growing}, single-layer mobile network \cite{michael2009maintaining,zavlanos2011graph,poonawala2015collision
,yang2010decentralized,martin2014multiagent}, multi-layer static network \cite{martin2014algebraic,shakeri2016maximizing}.  We focus on optimizing a new category of multi-layer MAS network connectivity in this work. Moreover, we design the network strategically by considering prevalent cyberattack models on the communication networks \cite{xu2006jamming}. 

Connectivity control of mobile robotic network has been often addressed either in a fully centralized way \cite{michael2009maintaining,hsieh2008maintaining} or completely decentralized way  \cite{simonetto2013constrained,yang2010decentralized}. 
Our proposed model stands in between these two frameworks, and thus results in balanced features in terms of resiliency and optimality. Some results of this paper have been presented in Chapter 3 of \cite{chengame_book}. This work extends it in multiple aspects, including more analytical results with complete proofs, thorough framework development, extensive discussions of the model and results, and new case studies.

\subsection{Organization of the Paper}
The rest of the paper is organized as follows. Section~\ref{s1} presents some basics of graph theory. Section~\ref{s2} formulates the secure multi-layer MAS network formation problem. We analyze the problem and present the meta-equilibrium solution concept in Section~\ref{analysis}. Section \ref{s3} reformulates the problem and develops an iterative algorithm to compute the solution. Section~\ref{s4} analyzes the strategy of attacker and another type of cyberattacks. Section~\ref{s5} uses some case studies to illustrate the results. Section~\ref{conclusion} concludes the paper.

\section{Background and Interdependent Networks }\label{s1}
In this section, we present the basics of interdependent networks and the algebraic connectivity metric.
\subsection{Preliminaries}
Let $G(V,E)$ be an undirected graph composed by $n$ nodes and interconnected by $m$ links. For a link $l$ that connects nodes $i$ and $j$ where the link weight equaling to $w_{ij}$, we define two vectors $\mathbf{a}_l$ and $\mathbf{b}_l$, where $\mathbf{a}_l{(i)}=1$, $\mathbf{a}_l{(j)}=-1$, $\mathbf{b}_l{(i)}=w_{ij}$, $\mathbf{b}_l{(j)}=-w_{ij}$, and all other entries 0. Then, the Laplacian matrix $\mathbf{L}$ of network $G$ can be expressed as 
\begin{equation}\label{laplacian2}
\mathbf{L}=\displaystyle\sum_{l=1}^m{\mathbf{a}_l}{\mathbf{b}_l^T},
\end{equation}
where ``$T$'' denotes the matrix transpose operator. Intuitively, the diagonal entry $\mathbf{L}_{ii}$ in the Laplacian matrix is equal to the weight of node $i$, i.e., $\mathbf{L}_{ii}=\sum_{j\in\mathcal{N}_i}w_{ij},\ \forall i\in V$, where $\mathcal{N}_i$ denotes the set of nodes that connects with node $i$. In addition, $\mathbf{L}_{ij}=-w_{ij},\ \forall i\neq j\in V,$ if nodes $i$ and $j$ are connected, and otherwise is 0. Additionally, Laplacian matrix is positive semidefinite, and $\mathbf{L}\mathrm{\mathbf{1}}=0$, where $\mathbf{1}$ is an $n$-dimensional vector with all one entries. Thus, by ordering the eigenvalues of $\mathbf{L}$ in an increase way, we obtain
\begin{equation}\label{lambdarelation}
0=\lambda_1\le \lambda_2\le...\le \lambda_n,
\end{equation}
where $\lambda_2(\mathbf{L})$ is called \textit{algebraic connectivity} (or Fiedler value) of $G$ \cite{fiedler1973algebraic}. Algebraic connectivity is an indicator of how well a graph is connected. Moreover, $\lambda_2(\mathbf{L})=0$, when $G$ is not connected. For a graph with Laplacian $\mathbf{L}$, the algebraic connectivity $\lambda_2(\mathbf{L})$ 
can be computed from the Courant-Fisher theorem \cite{horn2012matrix} as follows:
\begin{equation}\label{lambdainf}
\lambda_2(\mathbf{L})=\mathrm{min}\{{z^T\mathbf{L}z}|{z\in \mathbf{1}^\perp},||z||_2= 1 \},
\end{equation}
where $||\cdot||_2$ denotes the standard $L_2$ norm.
\subsection{Interdependent Networks}
For a two-layer interdependent network, we define two networks $G_1(V_1,E_1)$ and $G_2(V_2,E_2)$, where network 1 and network 2 are represented by $G_i$, for $i=1,2$, respectively. Network $i$, $i\in\{1,2\}$, is composed of $n_i=|V_i|$ nodes and $m_i=|E_i|$ links, where $|\cdot|$ denotes the cardinality of a set. The global network resulting from the connection of these two networks can be represented by $G=(V_1\ {\cup}\ V_2,\ E_1\ {\cup}\ E_2\ {\cup}\ E_{12})$, where $E_{12}$ is a set of \textit{intra-links} between $G_1$ and $G_2$. For convenience, we denote the network consisting of the intra-links between $G_1$ and $G_2$ as $G_{12}$. 
The adjacency matrix $\mathbf{A}$ of the global network $G$ has the entry $a_{ij}=w_{ij}$ if nodes $i$ and $j$ are connected and $a_{ij}=0$ otherwise.
Let $\mathbf{A}_1\in \mathbb{R}^{n_1\times n_1}$ and $\mathbf{A}_2\in \mathbb{R}^{n_2\times n_2}$ be the adjacency matrices of $G_1$ and $G_2$, respectively, and $n=n_1+n_2$. When $E_{12}\ne{\emptyset}$, the adjacency matrix $\mathbf{A}\in \mathbb{R}^{n\times n}$ takes the following form
$
\mathbf{A}=\begin{bmatrix}
\mathbf{A}_1&\mathbf{B}_{12}\\
\mathbf{B}_{12}^T&\mathbf{A}_2
\end{bmatrix},
$
where $\mathbf{B}_{12}\in \mathbb{R}^{n_1\times n_2}$ is a matrix capturing the effect of intra-links between networks. Define two diagonal matrices $\mathbf{D}_1\in \mathbb{R}^{n_1\times n_1}$ and $\mathbf{D}_2\in \mathbb{R}^{n_2\times n_2}$ as
$
\begin{cases}
(\mathbf{D}_1)_{ii}=\sum_{j}{(\mathbf{B}_{12})_{ij}},\\
(\mathbf{D}_2)_{ii}=\sum_{j}{(\mathbf{B}_{12}^T)_{ij}}.
\end{cases}
$
Based on
$
\mathbf{L}=\mathbf{D}-\mathbf{A}, 
$
we obtain the Laplacian matrix
$
\mathbf{L}=\begin{bmatrix}
\mathbf{L}_1+\mathbf{D}_1&-\mathbf{B}_{12}\\
-\mathbf{B}_{12}^T&\mathbf{L}_2+\mathbf{D}_2
\end{bmatrix},
$
where $\mathbf{L}_1$ and $\mathbf{L}_2$ are the Laplacians corresponding to $\mathbf{A}_1$ and $\mathbf{A}_2$, respectively.

\textit{Remark}: The above formulated two-layer MAS network can be easily extended to multi-layer scenarios.

\section{System Framework and Problem Formulation}\label{s2}
In this section, we introduce the system framework which includes the wireless communication model and the strategic interdependent MAS network formation.

\subsection{Wireless Communication Model}
In the MAS, we consider a set $V$ of robots in the network, and their positions at time $k$ are defined by the vector $\mathbf{x}(k)=\big(x_1(k);\ x_2(k);...;x_n(k)\big)\in{\mathbb{R}^{3 n}}$. 
Robots in the same network can exchange data via wireless communications. Denote the communication link between robots $i$ and $j$ as $(i,j)$. Then, the strength of the communication link $(i,j)$ is similar to the weight of the link in a network.  Thus, we associate a weight function
$
w: \mathbb{R}^3\times \mathbb{R}^3\rightarrow \mathbb{R}_+
$
with every communication link $(i,j)$, such that
\begin{equation}\label{wei}
\begin{split}
w_{ij}(k)&=w\big(x_i(k),x_j(k)\big)=f\left(\parallel
x_{ij}(k)\parallel_2^2\right),
\end{split}
\end{equation}
for some differentiable $f$ : $\mathbb{R}_+\rightarrow \mathbb{R}_+$, where $x_{ij}(k):=x_i(k)-x_j(k)$, and $\parallel
x_{ij}(k)\parallel_2$ is the distance between robots $i$ and $j$ at time $k$. To capture the communication strength decay with the distance, $f$ is a monotonically decreasing function.
A typical choice of $f$ is $f(d) = \delta^{(c_1-d)/(c_1-c_2)}$, where $\delta$, $c_1$ and $c_2$ are positive constants. Note that different forms of $f$ capture various decay rates of communication strength with distance \cite{tse2005fundamentals}. Thus, the weight of the link between robots is positive if their distance is within a threshold and degenerates to zero otherwise. Fig. \ref{Aij} shows an example of $f$ with $\delta=0.1$, $c_1=2$ and $c_2=6$.

\begin{figure}[t]
\centering
\includegraphics[width=0.65\columnwidth]{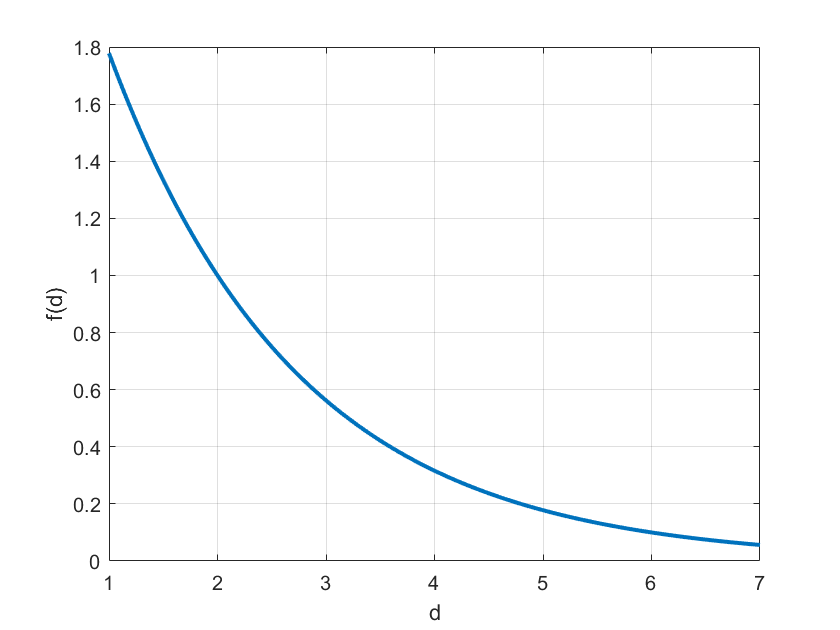}
\caption{Communication strength under function $f(d) = \delta^{(c_1-d)/(c_1-c_2)}$ with $\delta=0.1$, $c_1=2$ and $c_2=6$.}\label{Aij}
\end{figure}

\subsection{Secure Interdependent MAS Network Formation}
A two-layer MAS network model is shown in Fig. \ref{model}, where networks $G_1$ and $G_2$ include $n_1$ and $n_2$ number of robots, respectively. More generally, we label robots in $G_1$ as $1,2,...,n_1$, and robots in $G_2$ as $n_1+1,n_1+2,...,n_1+n_2$, i.e., ${V}_1:=\{1,2,...,n_1\}$ and ${V}_2:=\{n_1+1,n_1+2,...,n\}$. Note that $n=n_1+n_2$. Robots in these two layers can also communicate, and this kind of communication link is called \textit{intra-link} while the link inside of a network is known as \textit{inter-link}. 
The agents at two layers are interdependent, and thus the integrated MAS network can be modeled as a \textit{system-of-systems}.
\begin{figure}[t]
\centering
 \includegraphics[width=0.65\columnwidth]{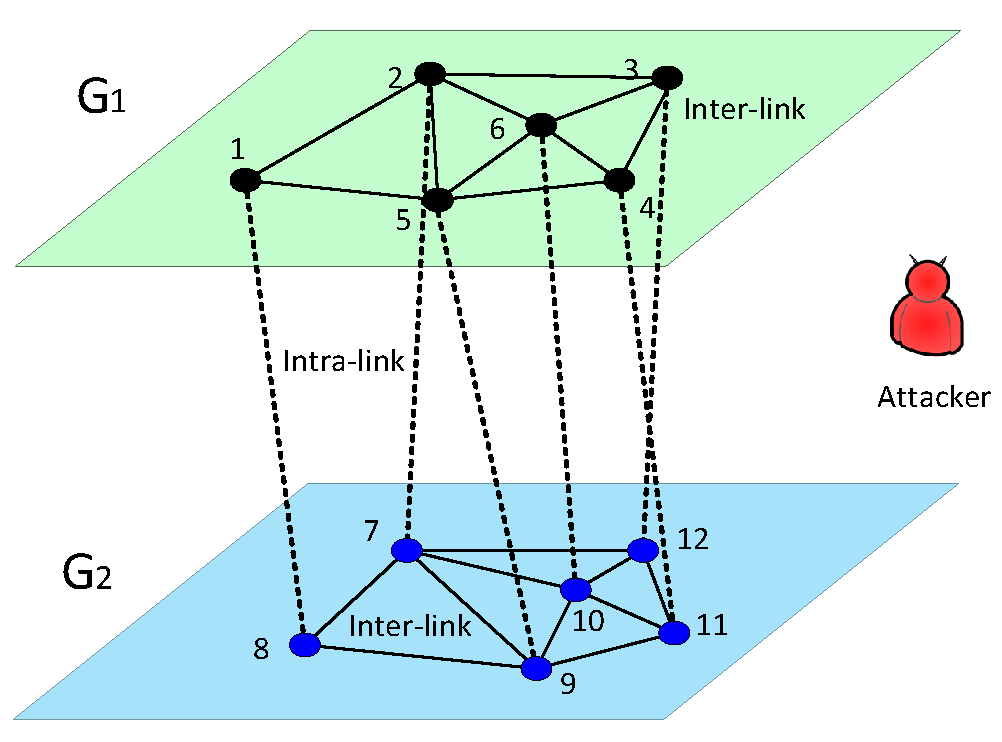}
\caption{Multi-layer MAS network in an adversarial environment.}\label{model}
\end{figure}

\subsubsection{Network Designer} We consider two players, player 1 ($P_1$) and player 2 ($P_2$), operating two interdependent MAS networks. $P_1$ controls robots in network $G_1$, and $P_2$ controls robots in $G_2$. Specifically, $P_1$ and $P_2$ update their own network with a fixed frequency by controlling the positions of robots. After each update, the communication link strength between robots are modified due to the change of distance. For simplicity, define $-\gamma:= \{1,2\}\setminus \gamma$, where $\gamma\in\{1,2\}$, and $\mathbf{x}:=(\mathbf{x}_1,\mathbf{x}_2)$, where $\mathbf{x}_1:=(x_1;...;x_{n_1})\in\mathbb{R}^{3 n_1}$ and $\mathbf{x}_2:=(x_{n_1+1};...;x_n)\in\mathbb{R}^{3 n_2}$. Specifically, $\mathbf{x}_1$ and $\mathbf{x}_2$ are decision variables denoting the position of robots in $G_1$ and $G_2$, respectively.  In addition, the action spaces of $P_1$ and $P_2$ are denoted by $\mathbf{X}_1$ and $\mathbf{X}_2$, respectively, which include all the possible network configurations. The set of pure strategy profiles $\mathbf{X}:=\mathbf{X}_1\times \mathbf{X}_2$ is the Cartesian product of the individual pure strategy sets. For each update, $P_\gamma$'s strategy $\mathbf{x}_\gamma$ is based on the current configuration of network $G_{-\gamma}$. The goal of both players is to optimize the SoS performance, i.e., maximize the algebraic connectivity of the global network $G$. Hence, the utility function for both players is $\lambda_2\big(\mathbf{L}_G(\mathbf{x})\big)$: $\mathbf{X}\rightarrow \mathbb{R}_+$, where $\mathbf{L}_G(\mathbf{x})$ is the Laplacian matrix of network $G$ when mobile robots have position $\mathbf{x}$.

In the adversarial network formation game, one of the constraints is the minimum distance between robots in each layer. Without this constraint, all robots at the same layer will converge to one point finally which is not a reasonable solution. Thus, we assign a minimum distance $\rho_1$ and $\rho_2$ for robots  in $G_1$ and $G_2$, respectively.

\subsubsection{Cyber Attacker} In addition to the network players $P_1$ and $P_2$, our framework also includes a malicious jamming attacker as shown in Fig. \ref{model}. The attacker is able to disrupt communication links via injecting a large amount of spam into the channel which leads to the link breakdown eventually because of overload of the link. The attacker's objective is to minimize the algebraic connectivity of the network through compromising links. Generally, the behavior of attacker is unknown to the network operators. Therefore, it is difficult for the network designers to make optimal strategies that can achieve the best performances of the network. However, by knowing that attackers are strategic and are more prone to disrupt the critical communication links in the network, the network operators can design a secure MAS network resistant to cyberattacks. Specifically, network designers first anticipate that the attacker can compromise a number $\psi\in\mathbb{N}^+$ of links, and then design the MAS network by taking into account the worst-case attack that leads to the most decrease of the network algebraic connectivity. The resistant property of the network is reflected by the fact that the designers prepare for the potential attacks and respond to them with best strategies. Here, $\psi$ quantifies the security level of the designed network. 

Denote $\mathcal{A}$ by the action space of the attacker which is the set including all the possible single communication link removal in the network. For convenience, we denote $\mathbf{L}^{e}_G(\mathbf{x})$ by the Laplacian matrix of the network after removing a set of links $e\subseteq \mathcal{A}$, i.e., the network after attack is $G(V,E_1\cup E_2\cup E_{12}\setminus{e})$, and the cardinality of $e$ is $|e| = \psi$ quantifying the ability of attacker, where $\psi\in\mathbb{N}^+$ is a positive integer. Denote the feasible set of $e$ by $\mathcal{E}$. Then, the cost function of the attacker can be captured by $\Lambda(\mathbf{x}_1,\mathbf{x}_2,e)\triangleq \lambda_2\big(\mathbf{L}^{e}_G(\mathbf{x})\big)$, for $\Lambda: \mathbf{X}_1\times \mathbf{X}_2\times \mathcal{E} \rightarrow \mathbb{R}_+$.

\subsection{Games-in-Games Formulation}\label{gamesingames}
During the MAS network formation, the interactions between two networks $G_1$ and $G_2$ can be modeled as a Nash game where both players aim to increase the global network connectivity. In addition, each network operator  plays a Stackelberg game with the malicious jamming attacker. Therefore, the multi-layer MAS network formation in the adversarial environment can be characterized by a games-in-games framework which is shown in Fig. \ref{games_in_games}. 
 In the following, we specifically formulate the attacker's and network operators' problems, respectively. 

\begin{figure}[t]
\centering
 \includegraphics[width=0.7\columnwidth]{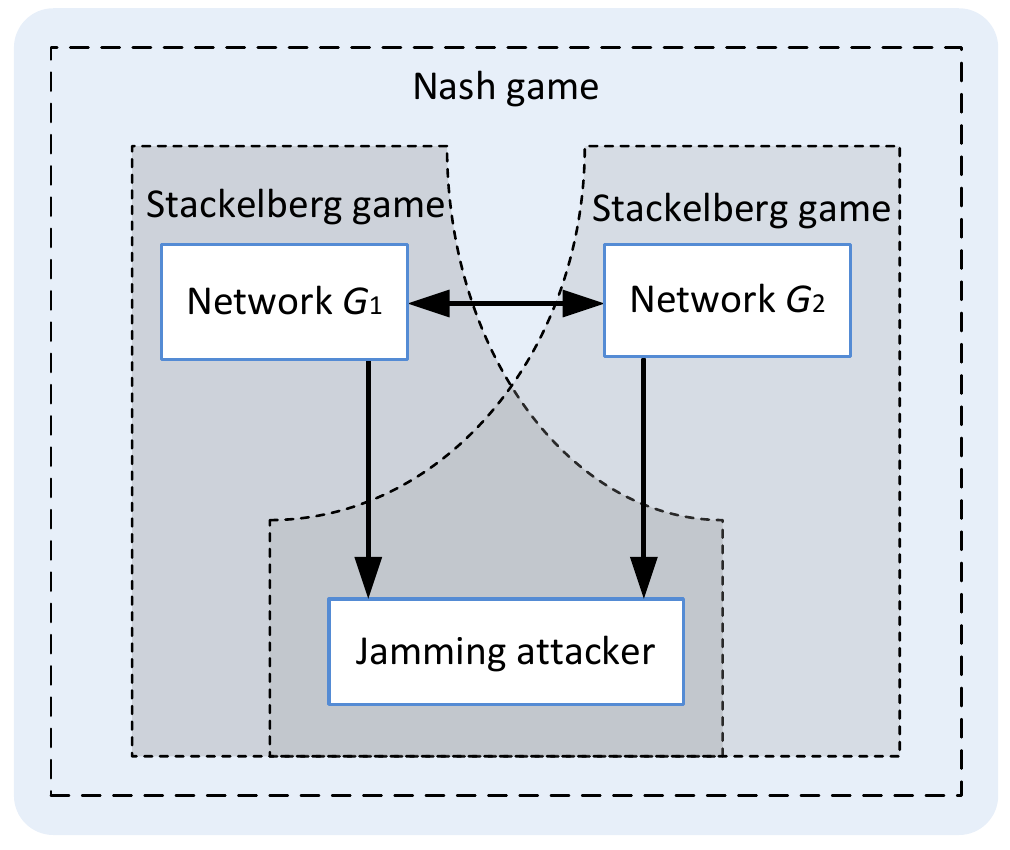} 
\caption{Games-in-Games framework which includes two network operators and one attacker. Both network operators prepare for the cyberattack which form a Stackelberg game with the attacker. In addition, two network operators are uncoordinated and aim to maximize the global network connectivity which create a Nash game.}\label{games_in_games}
\end{figure}

\subsubsection{Stackelberg Game}
In the Stackelberg game, network designer is the leader, and the jamming attacker is the follower.  The objective of the attacker is to minimize the algebraic connectivity of network $G$.  We can summarize the strategic behavior of the attacker into the following problem:
\begin{equation}
\mathcal{Q}_A^k:\ \ \ \min_{e\subseteq \mathcal{A},|e| = \psi}\  \lambda_2\big(\mathbf{L}^{e}_G(\mathbf{x}(k+1))\big).
\end{equation}

On the leader side, network operator $P_\gamma$ maximizes the algebraic connectivity of the network, where $\gamma\in\{1,2\}$, and his decision can be obtained via solving the optimization problem:
\begin{equation}\label{p1}
\begin{split}
\mathcal{Q}_\gamma^k:\ \ \ &\max_{\mathbf{x}_\gamma(k+c_\gamma)}\ \min_{e\subseteq \mathcal{A},|e| = \psi}\  \lambda_2\big(\mathbf{L}^e_G(\mathbf{x}(k+c_\gamma))\big)\\
\mathrm{s.t.}\ \ \ &||x_{ij}(k+c_\gamma)||_2\ge \rho_\gamma, \quad \forall (i,j)\in E_\gamma,\\
&||x_{ij}(k+c_\gamma)||_2\ge \rho_{12},\quad \forall i\in V_\gamma,\ \forall j\in V_{-\gamma},\\
&||x_{i}(k+c_\gamma)-x_{i}(k)||_2\leq d_\gamma,\quad \forall i\in V_\gamma,\\
& x_{j}(k+c_\gamma) = x_{j}(k),\ \forall  j\in V_{-\gamma}, 
\end{split}
\end{equation}
where $c_\gamma\in\mathbb{N}^+$ is a positive integer indicating the update frequency; $\rho_\gamma\in\mathbb{R}_+$ is the safety distance between robots; $\rho_{12}\in\mathbb{R}_+$ is the minimum distance between robots in different layers; and $d_\gamma\in\mathbb{R}_+$ is the maximum distance that robots in network $G_\gamma$ can move at each update. The constraint $x_{j}(k+c_\gamma) = x_{j}(k),\ j\in V_{-\gamma}$ captures the uncoordinated nature that each network operator can only control the robots at his layer. Furthermore, this constraint preserves security consideration between agents in $V_{-\gamma}$ and also ensures consistent connectivity improvement when player $\gamma$ updates his network.

The Stackelberg game between the attacker and network operator $\gamma$ can be represented by $\Xi_{\gamma}:=\{\mathcal{N}_{\gamma},\mathbf{X}_{\gamma},\mathcal{A},\lambda_2\}$ for $\gamma\in\{1,2\}$, where $\mathcal{N}_\gamma:=\{P_\gamma,Attacker\}$ is the set of players, $\mathbf{X}_\gamma$ and $\mathcal{A}$ are action spaces and $\lambda_2$ is the objective function.

\subsubsection{Nash Game}
The interaction between two robotic networks in an adversarial environment can be characterized as a Nash game in which both players aim to increase the global network connectivity. We denote this strategic game by $\Xi_I:=\{P_1,P_2,\mathbf{X}_1,\mathbf{X}_2,\lambda_2\}$. 

Note that the MAS network formation game is played repeatedly over time, and its structure is the same only with different initial conditions in terms of the robots' position. 
This two-person interdependent MAS network formation game can be naturally generalized into an $N$-person game where each player controls a subset of robots in the multi-layer networks.

\textit{Remark:} In the focused model, the attacker is the follower and is not assumed to act in a proactive way. The attacker is a player in the game whose strategic action is anticipated by the network designers. Under some other models, e.g., multi-stage sequential games where the attacker takes an action followed by the defender, the attacker plays a non-passive role and can manipulate the outcome by selecting an equilibrium. This case is a further extension of the model introduced in this paper. The methodologies that have introduced in this work can be extended and applied to this new context.

It is possible that the attacker could manipulate and exploit operator's anticipation for his own benefit. To capture this scenario, we need to propose an alternative model that extends the current framework. In this work, we assume that the network operators choose an action by anticipating credible adversarial behaviors. The attacker can manipulate the above defense mechanism by anticipating defender's strategy. One way to achieve this goal for the attacker is to make decisions over a time horizon instead of a single time step. When the time horizon contains two steps, then the attacker address a three-stage game (attacker-defender-attacker) that he takes the lead first. A dynamic game of similar three-stage structure has been investigated in \cite{alderson2011solving,chen2019dynamic}. Interested readers can refer to that for detailed description of the framework.

\section{Problem Analysis and Solution Concepts}\label{analysis}
In this section, we first reformulated problems in Section \ref{s2}, and then present the solution concept of the MAS network formation game. 

\subsection{Problem Reformulation}
Note that each network designer updates the robotic network iteratively based on the current configuration. It is essential to obtain the relationship between the updated position and the current one due to the natural dynamics of robots. To achieve this goal, we define $\mathcal{Z}_{ij}(k):=||x_{ij}(k)||_2^2$ for notational convenience. Analogous to applying Euler's first order method to continuous dynamics, we can obtain $\mathcal{Z}_{ij}(k+m)$ based on the current positions $x_i(k)$ and $x_j(k)$ as follows:
 \begin{equation}\label{discretedistance}
\begin{split}
{\mathcal{Z}}_{ij}&(k+m)+{\mathcal{Z}}_{ij}(k) \\
&= 2\{{x}_i(k+m)-{x}_j(k+m)\}^T\{x_i(k)-x_j(k)\}.
\end{split}
\end{equation}
Readers interested in the details of obtaining \eqref{discretedistance} can refer to Appendix \ref{appDerive}.
Similarly, by using the function $f$ in \eqref{wei}, the updated weight $w_{ij}(k+m)$ can be expressed as:
\begin{equation}\label{weight}
\begin{split}
w_{ij}(k+m)&=w_{ij}(k)+\frac{\partial f}{\partial ||x_{ij}||_2^2}\biggl\rvert_k (\mathcal{Z}_{ij}(k+m)-\mathcal{Z}_{ij}(k)).
\end{split}
\end{equation}
Therefore, we can obtain the Laplacian matrix $\mathbf{L}_G\big(\mathbf{x}(k+m)\big)$ by using \eqref{weight} for the global network.

Each network designer needs to solve a $\mathit{max}\ \mathit{min}$ problem which are not straightforward to deal with. We first present the following result.

\begin{theorem}\label{Thm1}
For a network containing $n$ nodes, the optimization problem 
\begin{equation}\label{thm12}
\max_{\mathbf{x}}\ \min_{e\subseteq \mathcal{A},|e|=\psi}\ \lambda_2(\mathbf{L}_G^e(\mathbf{x}))
\end{equation}
is equivalent to 
\begin{equation}\label{thm1}
\begin{split}
&\max_{\mathbf{x},\beta}\ \ \beta\\
&\mathrm{s.t.}\ \ \ \mathbf{L}_G^e(\mathbf{x})\succeq \beta (\mathbf{I}_{n}-\frac{1}{n}{\mathbf{1}\mathbf{1}}^T),\ \forall e\subseteq \mathcal{A},\ |e|=\psi,
\end{split}
\end{equation}
where $\beta$ is a scalar, and $\mathbf{I}_n$ is an $n$-dimensional identity matrix. Note that the optimal $\mathbf{x}$ and the corresponding objective values in these two problems are equal.
\end{theorem}

\begin{proof}
Let $v_i$ be the eigenvector associated with eigenvalue $\lambda_i$ of the Laplacian matrix $\mathbf{L}_G^e(\mathbf{x})$, for $\forall i\in V$. Since $\mathbf{L}_G^e(\mathbf{x})$ is real and symmetric, its eigenvectors can be chosen such that they are real and orthonormal, i.e., $v_i^T v_j=0, \forall i\neq j\in V$ and $v_i^T v_i=1$. Specially, we define $v_1:=\frac{\mathbf{1}}{\sqrt{n}}$, which is actually the eigenvector corresponding to $\lambda_1=0$. For convenience, we denote $\mathbf{V}:=[v_1,v_2,...,v_n]$, and $\Upsilon = diag(\lambda_1,\lambda_2,...,\lambda_n)$, where $diag$ is the diagonal operator. Thus, $\Upsilon$ is a diagonal matrix with $i$th diagonal entry $\lambda_i(\mathbf{L}_G^e(\mathbf{x}))$. Furthermore, based on the orthonormal basis, we obtain $\mathbf{VV}^T = \sum_{i=1}^n  v_i v_i^T = \mathbf{I}_{n}$.  The definition of eigenvalue yields $\mathbf{L}_G^e(\mathbf{x}) \mathbf{V} =  \mathbf{V}\Upsilon$. Multiplying  by $\mathbf{V}^T$ on the right leads to eigen-decomposition:
\begin{equation}\label{spectral}
\mathbf{L}_G^e(\mathbf{x})= \mathbf{L}_G^e(\mathbf{x}) \mathbf{VV}^T = \mathbf{V}\Upsilon \mathbf{V}^T =  \sum_{i=1}^{n}\lambda_{i}(\mathbf{L}_G^e(\mathbf{x}))v_i v_i^T.
\end{equation}
Since $\lambda_1=0$, equation \eqref{spectral} can be simplified as
\begin{equation}\label{simplespec}
\mathbf{L}_G^e(\mathbf{x})=\sum_{i=2}^{n}\lambda_{i}(\mathbf{L}_G^e(\mathbf{x}))v_i v_i^T.
\end{equation}
Next, we add $\lambda_2(\mathbf{L}_G^e(\mathbf{x}))v_1 v_1^T$ to both sides of \eqref{simplespec} and obtain
$
\mathbf{L}_G^e(\mathbf{x})+\lambda_2(\mathbf{L}_G^e(\mathbf{x}))v_1 v_1^T=\sum_{i=2}^{n}\lambda_{i}(\mathbf{L}_G^e(\mathbf{x}))v_i v_i^T+\lambda_2(\mathbf{L}_G^e(\mathbf{x}))v_1 v_1^T.
$ Note that similar to \eqref{spectral}, $ \sum_{i=2}^n (\lambda_i(\mathbf{L}_G^e(\mathbf{x})) - \lambda_2(\mathbf{L}_G^e(\mathbf{x}))) v_iv_i^T + 0\cdot v_1v_1^T= \mathbf{V}(\Upsilon-diag(0,\lambda_2,\lambda_2,...,\lambda_2))\mathbf{V}^T = \mathbf{L}_G^e(\mathbf{x}) - diag(0,\lambda_2,\lambda_2,...,\lambda_2)$, where we use $\lambda_2$ for clarity in the diagonal entries. The eigenvalues of matrix $\mathbf{L}_G^e(\mathbf{x}) - diag(0,\lambda_2,\lambda_2,...,\lambda_2)$ can be obtained straightforwardly as $\{0,0,\lambda_3-\lambda_2,...,\lambda_n-\lambda_2\}$ and thus $\mathbf{L}_G^e(\mathbf{x}) - diag(0,\lambda_2,\lambda_2,...,\lambda_2)\succeq 0$ based on property \eqref{lambdarelation}, meaning that it is positive semidefinite.
In sum, $\sum_{i=2}^n (\lambda_i(\mathbf{L}_G^e(\mathbf{x})) - \lambda_2(\mathbf{L}_G^e(\mathbf{x}))) v_iv_i^T\succeq 0$ which can be rewritten as $\mathbf{L}_G^e(\mathbf{x}) \succeq \sum_{i=2}^{n}\lambda_{2}(\mathbf{L}_G^e(\mathbf{x}))v_i v_i^T$. Then,
\begin{equation}
\begin{split}
&\mathbf{L}_G^e(\mathbf{x})+\lambda_2(\mathbf{L}_G^e(\mathbf{x}))v_1 v_1^T\succeq \sum_{i=2}^{n}\lambda_{2}(\mathbf{L}_G^e(\mathbf{x}))v_i v_i^T\\&+
\lambda_2(\mathbf{L}_G^e(\mathbf{x}))v_1 v_1^T = \lambda_{2}(\mathbf{L}_G^e(\mathbf{x}))\sum_{i=1}^{n}v_i v_i^T.
\end{split}
\end{equation}
Thus, we obtain $\mathbf{L}_G^e(\mathbf{x})\succeq \lambda_{2}(\mathbf{L}_G^e(\mathbf{x}))(\mathbf{I}_{n}-v_1 v_1^T)$, yielding
\begin{align}
\mathbf{L}_G^e(\mathbf{x})\succeq \lambda_{2}(\mathbf{L}_G^e(\mathbf{x}))(\mathbf{I}_{n}-\frac{1}{n}\mathbf{1}\mathbf{1}^T).\label{semi}
\end{align}
The above analysis is for any given attacker's strategy $e\subseteq \mathcal{A}$. Next, we show that our modified algebraic connectivity maximization problem is equivalent to $\max_{\mathbf{x},\beta}\ \beta$ in \eqref{thm1}, i.e., $\max_{\mathbf{x},\beta}\ \beta=\lambda_2(\mathbf{L}_G^{e^*}(\mathbf{x}^*))$, where $\mathbf{x}^*$ and $e^*$ are the optimal decisions. For convenience, we denote $\beta^*=\max_{\mathbf{x},\beta}\ \beta$. The proof includes two parts. First, we show that $\beta^* \geq \lambda_2(\mathbf{L}_G^{e^*}(\mathbf{x}^*))$. We aim to maximize the algebraic connectivity $\lambda_2(\mathbf{L}_G^e(\mathbf{x}))$, and $(\mathbf{x}^*,e^*)$ is a feasible solution pair. Therefore, based on \eqref{semi}, $\beta^* \geq \lambda_2(\mathbf{L}_G^{e^*}(\mathbf{x}^*))$ should hold, since $\beta$ is a free variable to optimize in the problem \eqref{thm1} while its counterpart in \eqref{semi}, $\lambda_{2}(\mathbf{L}_G^e(\mathbf{x}))$, is dependent on $\mathbf{x}$ and $e$. Second, we show that $\beta^* \leq \lambda_2(\mathbf{L}_G^{e^*}(\mathbf{x}^*))$. Since $\beta^*,e^*,\mathbf{x}^*$ are feasible, then, the constraints in \eqref{thm1} should be satisfied, i.e., ${L}_G^e(\mathbf{x}^*)\succeq \beta^*(\mathbf{I}_{n}-\frac{1}{n}\mathbf{1}\mathbf{1}^T),\ \forall e\subseteq \mathcal{A},\ |e|=\psi$, which gives
$
{L}_G^{e^*}(\mathbf{x}^*)\succeq \beta^*(\mathbf{I}_{n}-\frac{1}{n}\mathbf{1}\mathbf{1}^T).
$
Let $\mu$ be any unit vector that satisfies $\mu^T v_1=0.$ Then, we obtain
$
\mu^T\mathbf{L}_G^{e^*}(\mathbf{x}^*)\mu\geq \mu^T\beta^*(\mathbf{I}_{n}-\frac{1}{n}\mathbf{1}\mathbf{1}^T)\mu
\rightarrow\ \mu^T\mathbf{L}_G^{e^*}(\mathbf{x}^*)\mu\geq \beta^*\mu^T \mathbf{I}_{n}\mu-\beta^*\mu^Tv_1 v_1^T\mu
\rightarrow\ \mu^T\mathbf{L}_G^{e^*}(\mathbf{x}^*)\mu\geq \beta^*\mu^T \mathbf{I}_{n}\mu=\beta^*.
$
Since vector $\mu$ is not fixed, and based on \eqref{lambdainf}, we have $\lambda_2(\mathbf{L}_G^{e^*}(\mathbf{x}^*))\geq \beta^*$. Therefore,  $\max_{\mathbf{x},\beta}\beta=\lambda_2(\mathbf{L}_G^{e^*}(\mathbf{x}^*))$, and \eqref{thm12} is equivalent to \eqref{thm1}. 
\end{proof}

Next, we define a new Stackelberg game $\widetilde{\Xi}_\gamma:=\{\mathcal{N}_{\gamma},\mathbf{X}_{\gamma},\mathcal{A},\alpha_\gamma,\lambda_2\}$, for $\gamma\in\{1,2\}$, where $\mathcal{N}_\gamma$, $\mathbf{X}_\gamma$ and $\mathcal{A}$ are the same as those defined in game ${\Xi}_\gamma$; $\alpha_\gamma$ and $\lambda_2$ are the objective functions of the network designer and attacker, respectively. Based on  \eqref{discretedistance}, \eqref{weight} and Theorem \ref{Thm1}, the network designer $\gamma$'s problem is formulated as follows, for $\gamma\in\{1,2\}$:
\begin{equation*}
\begin{split}
\widetilde{\mathcal{Q}}_\gamma^k:\ \ \ &\max_{\mathbf{x}_\gamma(k+c_\gamma),\alpha_\gamma(k+c_\gamma)}\qquad \alpha_\gamma(k+c_\gamma)\\
\mathrm{s.t.}\ \ \ &\mathbf{L}_G^e(k+c_\gamma)\succeq \alpha_\gamma(k+c_\gamma) (\mathbf{I}_{n}-\frac{1}{n}{\mathbf{1}\mathbf{1}}^T),\\
&\qquad\qquad\qquad \forall e\subseteq \mathcal{A},\ |e|=\psi,\\
&2\{{x}_i(k+c_\gamma)-{x}_j(k+c_\gamma)\}^T\{x_i(k)-x_j(k)\}\\
&\qquad\qquad={\mathcal{Z}}_{ij}(k+c_\gamma)+{\mathcal{Z}}_{ij}(k),\quad \forall i,j\in V_\gamma,\\
&||x_{ij}(k+c_\gamma)||_2\ge \rho_\gamma,\ \ \ \ \ \forall (i,j)\in E_\gamma,\\
&||x_{ij}(k+c_\gamma )||_2\ge \rho_{12},\quad \forall i\in V_\gamma,\ \forall j\in V_{-\gamma},\\
&||x_{i}(k+c_\gamma)-x_{i}(k)||_2\leq  d_\gamma,\quad \forall i\in V_\gamma,\\
& x_{j}(k+c_\gamma) = x_{j}(k),\ \forall j\in V_{-\gamma}.
\end{split}
\end{equation*}
Note that Laplacian matrices $\mathbf{L}_G^e(k+c_\gamma)$, for $\gamma=1,2$, are constructed based on \eqref{weight}.

The above analysis leads to the following corollary.
\begin{corollary}\label{equi}
The Stackelberg game $\widetilde{\Xi}_\gamma$ is strategically equivalent to the game ${\Xi}_\gamma$ defined in Section \ref{gamesingames}, for $\gamma\in\{1,2\}$. The interactions between two network operators can be captured by a strategic equivalent Nash game denoted by $\widetilde{\Xi}_{I}$, where $\widetilde{\Xi}_{I}$ includes $\alpha_\gamma$, $\gamma =1,\ 2$..
\end{corollary}

\subsection{Equilibrium Solution Concepts}

\subsubsection{Stackelberg Equilibrium of the Adversarial Game $\widetilde\Xi_{\gamma}$}
In the Stackelberg game, the attacker's strategy is the best response to the action that network designer chooses. Recall that $\Lambda(\mathbf{x}_1,\mathbf{x}_2,e)= \lambda_2\big(\mathbf{L}^{e}_G(\mathbf{x})\big)$, and the formal definition of best response is as follows.
\begin{definition}[Best Response]
For a given strategy pair $(\mathbf{x}_1,\mathbf{x}_2)$, where $\mathbf{x}_1\in\mathbf{X}_1$ and $\mathbf{x}_2\in\mathbf{X}_2$, the best response of the attacker is defined by
$
BR(\mathbf{x}_1,\mathbf{x}_2):=\{e':\Lambda(\mathbf{x}_1,\mathbf{x}_2,e')\leq \Lambda(\mathbf{x}_1,\mathbf{x}_2,e),
 \forall e,e'\subseteq\mathcal{A},\ |e'|=|e|=\psi\}.
$
\end{definition}

Thus, we give the definition of the Stackelberg equilibrium of game $\widetilde\Xi_{\gamma}$, for $\gamma\in\{1,2\}$.

\begin{definition}[Stackelberg Equilibrium]
For a given $\mathbf{x}_{-\gamma}\in\mathbf{X}_{-\gamma}$, the profile $(\mathbf{x}_\gamma^*,e^*)$ constitutes a Stackelberg equilibrium of the adversarial game $\widetilde\Xi_{\gamma}$, for $\gamma\in\{1,2\}$, if the following conditions are satisfied:
\begin{enumerate}
\item Attacker's strategy $e^*\subseteq\mathcal{A}$, where $|e^*|=\psi$, is a best response to $(\mathbf{x}_\gamma^*,\mathbf{x}_{-\gamma})$, i.e.,
$e^*\in BR(\mathbf{x}_\gamma^*,\mathbf{x}_{-\gamma}).$
\item Network designer $\gamma$'s strategy $\mathbf{x}_\gamma^*\in\mathbf{X}_\gamma$ satisfies
\begin{align*}
&\min_{e\in BR(\mathbf{x}_\gamma^*,\mathbf{x}_{-\gamma})}\Lambda(\mathbf{x}_\gamma^*,\mathbf{x}_{-\gamma},e)=\\
&\max_{\mathbf{x}_\gamma\in\mathbf{X}_\gamma}\ \min_{e\in BR(\mathbf{x}_\gamma,\mathbf{x}_{-\gamma})} \Lambda(\mathbf{x}_\gamma,\mathbf{x}_{-\gamma},e) \triangleq \Lambda^{\gamma*},
\end{align*}
where $\Lambda^{\gamma*}$ is the Stackelberg utility of the designer $\gamma$.
\end{enumerate}
\end{definition}

\subsubsection{Nash Equilibrium of the MAS Network Formation Game $\widetilde\Xi_{I}$}
After $P_1$ takes his action at step $k$, $G_1$ and $G_{12}$ are reconfigured, where $G_{12}$ is the network between $G_1$ and $G_2$. We denote network $G_1$ and $G_{12}$ at stage $k$ as $G_{1,k}$ and $G_{12,k}$, respectively. For simplicity, we further define $\widetilde G_{12,k}:=G_{1,k} \cup G_{12,k}$, which is a shorthand notation for the merged network. Then, network $G_k$ can be expressed as $G_k= \widetilde G_{12,k}\ \cup G_{2,k}$. Similarly, after $P_2$ updates network $G_2$ at step $k$, the whole network $G_k$ becomes $ G_k=\widetilde G_{21,k}\ \cup G_{1,k}$, where $\widetilde G_{21,k}:=G_{2,k} \cup G_{12,k}$. Then, the formal definition of Nash equilibrium (NE) which depends on the \textit{position} of robots is as follows.

\begin{definition}[Nash Equilibrium]\label{NE_definition}
The Nash equilibrium solution to game $\widetilde\Xi_{I}$ is a strategy profile $\mathbf{x}^*$, where $\mathbf{x}^*=(\mathbf{x}_1^*,{\mathbf{x}_2^*})\in{\mathbf{X}}$, that satisfies
\begin{equation*}\label{NE_discrete}
\begin{split}
\lambda_2\big(\mathbf{L}_{ G_{k}}(\mathbf{x}_1^*,\mathbf{x}_2^*) \big)\ge \lambda_2\big(\mathbf{L}_{G_{k}}(\mathbf{x}_1,\mathbf{x}_2^*) \big),\\
\lambda_2\big(\mathbf{L}_{ G_{k}}(\mathbf{x}_1^*,\mathbf{x}_2^*) \big)\ge \lambda_2\big(\mathbf{L}_{G_{k}}(\mathbf{x}_1^*,\mathbf{x}_2) \big),
\end{split}
\end{equation*}
for $\forall \mathbf{x}_1\in \mathbf{X}_1$ and $\forall \mathbf{x}_2 \in \mathbf{X}_2$, where $k$ denotes the time step.
\end{definition}
Note that $\mathbf{L}_{ G_{k}}$ in Definition \ref{NE_discrete} captures the network characteristic under all possible attacks instead of a particular one. At the NE point, no player can individually increase the global network connectivity by reconfiguring their MAS network.

\subsubsection{Meta-Equilibrium of the Games-in-Games} 
To design a secure multi-layer MAS, each network operator should take into account the attacker's behavior and the other network operator's strategy. The Nash game and the Stackelberg game are inherently coupled, as the adversarial consideration naturally affects the operators' decisions in the Nash game. Furthermore, the NE (Definition \ref{NE_definition}) alone cannot fully capture all the required elements in the games-in-games framework, as the adversarial strategy is characterized by the Stackelberg equilibrium.  Hence, a holistic solution concept is necessary for our composed game which is presented as follows.
\begin{definition}[Meta-Equilibrium]
The meta-equilibrium of the multi-layer MAS network formation game is captured by the profile $(\mathbf{x}_1^*,\mathbf{x}_2^*,e^*)$ which satisfies the following conditions:
\begin{enumerate}
\item $(\mathbf{x}_\gamma^*,e^*)$ constitutes a Stackelberg equilibrium of game $\widetilde\Xi_{\gamma}$, for $\gamma = 1,2$.
\item $\mathbf{x}^*=(\mathbf{x}_1^*,{\mathbf{x}_2^*})$ is an NE of game $\widetilde\Xi_{I}$.
\end{enumerate}
\end{definition}

Note that the meta-equilibrium is a solution concept for the composed game (games-in-games), which captures the incoordination between two network operators (Nash game) and the security consideration of each network operator (Stackelberg game) simultaneously.

\section{SDP-Based Approach and Iterative Algorithm}\label{s3}

In this section, we reformulate the network designer's problem as a semidefinite programming (SDP) and design an algorithm to compute the meta-equilibrium of the  MAS network formation game.

\subsection{SDP Reformulation}
Notice that in $\widetilde{\mathcal{Q}}_\gamma^k$, the minimum distance constraints $||x_{ij}(k+c_\gamma)||_2\geq \rho_\gamma,\ \forall (i,j)\in E_\gamma$, are \textit{nonconvex}. To address this issue, we regard $\mathcal{Z}_{ij}(k+c_\gamma)$ as a new decision variable. Based on the definition $\mathcal{Z}_{ij}(t):=||x_{ij}(t)||_2^2$, we have $||x_{ij}(k+c_\gamma)||_2^2=\mathcal{Z}_{ij}(k+c_\gamma)$. Note that the Laplacian matrix $\mathbf{L}_G^e(k+c_\gamma)$ depends linearly on $\mathcal{Z}_{ij}(k+c_\gamma)$, $i,j\in V$,  based on \eqref{weight}. Then, we solve problems $\widetilde{\mathcal{Q}}_\gamma^k$ with respect to unknowns $\mathcal{Z}_{ij}(k+c_\gamma)$ and $\mathbf{x}(k+c_\gamma)$ jointly. In this way, $\widetilde{\mathcal{Q}}_\gamma^k$ becomes a convex problem. However, due to the coupling between the robots position and the distance vectors, solving $\widetilde{\mathcal{Q}}_\gamma^k$ via merely adding new variables yields inconsistency between the obtained solutions $\mathbf{x}(k+c_\gamma)$ and $\mathcal{Z}_{ij}(k+c_\gamma),\ \forall i,j\in V$. To address this issue, we first present the following definition.
\begin{definition}[Euclidean Distance Matrix]
Given the positions of a set of $n$ points denoted by $\{x_1,...,x_n\}$, the Euclidean distance matrix representing the points spacing is defined as $$\mathbf{D}:=[d_{ij}]_{i,j=1,...,n},\ \mathrm{where}\ d_{ij}=||x_i-x_j||_2^2.$$
\end{definition}
A critical property of the Euclidean distance matrix is summarized in the following theorem.

\begin{theorem}[\cite{dattorro2010convex}]\label{EDM}
A matrix $\mathbf{D}=[d_{ij}]_{i,j=1,...,n}$ is an Euclidean distance matrix if and only if 
\begin{equation}\label{sdp_EDM}
-\mathbf{CDC}\succeq 0,\ \mathrm{and}\ d_{ii}=0,\ i=1,...,n,
\end{equation}
where $\mathbf{C}:=\mathbf{I}_n - \frac{1}{n}\mathbf{11}^T$.
\end{theorem}

Note that \eqref{sdp_EDM} is a necessary and sufficient condition that ensures $\mathbf{D}$ an  Euclidean distance matrix. In addition, the inequality and equality in \eqref{sdp_EDM} are both convex. Therefore, Theorem \ref{EDM} provides an approach to avoid the inconsistency between the robots position and distance vectors when they are treated as independent variables. In specific, denote $\mathbf{Z}=[\mathcal{Z}_{ij}]_{i,j\in V}$, $\mathbf{C}=\mathbf{I}_{n} - \frac{1}{n}\mathbf{11}^T$, and we can further reformulate problems $\widetilde{\mathcal{Q}}_\gamma^k$, $\gamma\in\{1,2\}$, as
\begin{equation}\label{ap2}
\begin{split}
\overline{\mathcal{Q}}_\gamma^k:\ \ \ &\max_{\mathbf{x}_\gamma(k+c_\gamma),\mathbf{Z}(k+c_\gamma),\alpha_\gamma(k+c_\gamma)}\quad \alpha_\gamma(k+c_\gamma)\\
\mathrm{s.t.}\ \ \ &\mathbf{L}_G^e(k+c_\gamma)\succeq \alpha_\gamma(k+c_\gamma) \mathbf{C},\quad \forall e\subseteq \mathcal{A},\ |e|=\psi,\\
&2\{{x}_i(k+c_\gamma)-{x}_j(k+c_\gamma)\}^T\{x_i(k)-x_j(k)\}\\
&\qquad\qquad={\mathcal{Z}}_{ij}(k+c_\gamma)+{\mathcal{Z}}_{ij}(k),\quad \forall i,j\in V_\gamma,\\
&\mathcal{Z}_{ij}(k+c_\gamma)\ge \rho_\gamma^2,\ \forall (i,j)\in E_\gamma,\\
&\mathcal{Z}_{ij}(k+c_\gamma)\ge \rho_{12}^2,\quad \forall i\in V_\gamma,\ \forall j\in V_{-\gamma},\\
& -\mathbf{CZ}(k+c_\gamma)\mathbf{C}\succeq 0,\ \mathcal{Z}_{ii}(k+c_\gamma)=0,\ i\in V,\\
&||x_{i}(k+c_\gamma)-x_{i}(k)||_2\leq  d_\gamma,\ \forall i\in V_\gamma,\\
& x_{j}(k+c_\gamma) = x_{j}(k),\ \forall j\in V_{-\gamma}.
\end{split}
\end{equation}

\textit{Remark}: In $\overline{\mathcal{Q}}_\gamma^k$, the relationship $||x_{ij}(k+c_\gamma)||_2^2=\mathcal{Z}_{ij}(k+c_\gamma)$ is ensured by the constraint $2\{{x}_i(k+c_\gamma)-{x}_j(k+c_\gamma)\}^T\{x_i(k)-x_j(k)\}={\mathcal{Z}}_{ij}(k+c_\gamma)+{\mathcal{Z}}_{ij}(k)$. In this constraint, ${\mathcal{Z}}_{ij}(k)$ is known which can be calculated based on the current position $\mathbf{x}(k)$, i.e., $\mathcal{Z}_{ij}(k) =
||x_{i}(k) - x_{j}(k)||_2^{2}$. Furthermore, constraints $-\mathbf{CZ}(k+c_\gamma)\mathbf{C}\succeq 0$ and $\mathcal{Z}_{ii}(k+c_\gamma)=0$ guarantee that the elements in $\mathbf{Z}(k+c_\gamma)$ are equal to the distances between corresponding nodes.

Note that $P_\gamma$ controls robots in $G_\gamma$ and reconfigures the network by solving $\overline{\mathcal{Q}}_\gamma^k$ to obtain the new positions of robots for $\gamma\in\{1,2\}$.  Furthermore, $\overline{\mathcal{Q}}_\gamma^k$ becomes convex and admits an SDP formulation which can be solved efficiently.

\subsection{Iterative Algorithm}\label{iterative_algorithm}
We have obtained the SDP formulations $\overline Q_\gamma^k$, $\gamma=1,2$, and next we aim to find the solution that results in a meta-equilibrium MAS configuration. In the MAS formation game, $P_1$ controls robots in $G_1$ and reconfigures the network by solving the optimization problem $\overline Q_1^k$ to obtain a new position of each robot. $P_2$  controls robots in network $G_2$ in a similar way by solving $\overline Q_2^k$. Note that the players' action at the current step can be seen as a best-response to the network at the previous step by taking the worst-case attack into account.  
Since both players maximize the global network connectivity at every update step, then one approach to find the meta-equilibrium solution is to address $\overline Q_1^k$ and $\overline Q_2^k$ iteratively by two players until the yielding MAS possesses the same secure topology, i.e., $P_1$ and $P_2$ cannot increase the network connectivity further through reallocating their robots. For clarity, Algorithm \ref{algorithm3} shows the updating details.
A typical example of the algorithm is \textit{alternating update} in which $P_1$ and $P_2$ have the same update frequency but not update at the same time and reconfigure the MAS network sequentially. 
 
\subsection{Structural Results}
Regarding the feasibility of the problems $\overline{\mathcal{Q}}_\gamma^k$ for $\gamma=1,2$, we have the following remark.

\textit{Remark:}
For a given multi-layer MAS network where the distance between robots satisfies the predefined minimum distance constraint, problems $\overline{\mathcal{Q}}_1^k$ and $\overline{\mathcal{Q}}_2^k$ are always feasible.
The feasibility can be indeed achieved by the players' strategies  that they do not update the position of robots at step $k$.

\begin{algorithm}[!t]
\caption{Secure and resilient MAS  network formation}\label{algorithm3}
\begin{algorithmic}[1]
\State Initialize mobile robots' position $x_i(0),\ \forall i\in V$, $\mathbf{x}(1)=2\mathbf{x}(1-c_1)=2\mathbf{x}(1-c_2)$, $c_1$, $c_2$, $\kappa=10^{-6}$.
\For {$k=1,2,3,...$}
\If{$k\mod c_1=0$ and $\lVert \mathbf{x}(k) - \mathbf{x}(k-c_1) \rVert_\infty > \kappa$}
\State $P_1$ obtains new strategy $\mathbf{x}_1(k+c_1)$ via solving $\overline{\mathcal{Q}}_1^k$ \Else\ $\mathbf{x}_1(k) = \mathbf{x}_1(k-1)$
\EndIf
\If {$k\mod c_2=0$ and $\lVert \mathbf{x}(k) - \mathbf{x}(k-c_2) \rVert_\infty > \kappa$}
\State $P_2$ obtains new strategy $\mathbf{x}_2(k+c_2)$ via solving $\overline{\mathcal{Q}}_2^k$
\Else\  $\mathbf{x}_2(k) = \mathbf{x}_2(k-1)$
\EndIf
\State \textbf{Break} if $\lVert \mathbf{x}(k) - \mathbf{x}(k-c_1) \rVert_\infty \leq \kappa$ and $\lVert \mathbf{x}(k) - \mathbf{x}(k-c_2) \rVert_\infty \leq \kappa$ and $k>\max(c_1,c_2)$
\State $k\gets k+1$
\EndFor
\State \textbf{return} $\mathbf{x}(k)$
\end{algorithmic}
\end{algorithm}

When $\overline{\mathcal{Q}}_1^k$ and $\overline{\mathcal{Q}}_2^k$ are feasible at each update step, another critical property is the convergence of the proposed iterative algorithm. The result is summarized in Theorem \ref{convergence}.

\begin{theorem}\label{convergence}
The proposed Algorithm \ref{algorithm3} of the adversarial network formation game converges to a meta-equilibrium asymptotically.
\end{theorem}

\begin{proof}
Note that at each stage of play, the game captured by $\mathcal{\bar Q}_1^k$ and $\mathcal{\bar Q}_2^k$ can be characterized as a constrained \textit{potential game} due to the identical objective of two players \cite{monderer1996potential}. The designed algorithm is based on the best response dynamics  which yields a non-decreasing network connectivity sequence. Furthermore, for a network with $n$ nodes, its algebraic connectivity is upper bounded by a value depending on $f(d_\gamma)$  \cite{godsil2013algebraic}. Thus, based on the monotone convergence theorem \cite{ganter2012formal}, the algorithm converges to a meta-equilibrium asymptotically.
\end{proof}

We next show that designers' adversary-anticipation is beneficial to the network performance and the multi-layer MAS network is resistant to strategic attacks.
\begin{lemma}
Under the established games-in-games model in Section \ref{gamesingames}, we obtain $\lambda_2\big(\mathbf{L}^{e_k}_{ G_{k}}(\mathbf{x}_1(k),\mathbf{x}_2(k)) \big)\geq \lambda_2\big(\mathbf{L}^{\tilde{e}_k}_{ G_{k}}(\tilde{\mathbf{x}}_1(k),\tilde{\mathbf{x}}_2(k)) \big)$, where $(\mathbf{x}_1(k),\mathbf{x}_2(k))$ is a strategy pair of network operators resulting from \eqref{ap2} with an appropriate time index; $(\tilde{\mathbf{x}}_1(k),\tilde{\mathbf{x}}_2(k))$ is a corresponding strategy pair without adversary-anticipation (ignoring the first constraint in \eqref{ap2}); and $e_k$ and $\tilde{e}_k$ are the attacker's admissible strategies at step $k$ satisfying $e_k\in BR(\mathbf{x}_1(k),\mathbf{x}_2(k))$ and $\tilde{e}_k\in BR(\tilde{\mathbf{x}}_1(k),\tilde{\mathbf{x}}_2(k))$. The inequality also holds at the meta-equilibrium $(\mathbf{x}^*,e^*)$.
\end{lemma}
\begin{proof}
We prove the result consecutively in two parts, i.e.,  before and after the network reaching meta-equilibrium. For the first part, the network configuration at every time step $k$ incorporates the consideration of $|e|=\psi$ link attacks and the physical constraints, i.e., the maximum distance that robot can move at one time step and the minimum distances between robots. Assume that designer $\gamma$, $\gamma\in\{1,2\}$, updates the configuration and the position of robots becomes $\mathbf{x}_{\gamma}(k)$ at time $k$. Then, the updated position $\mathbf{x}_{\gamma}(k)$ of player $\gamma$ is a best response to $\mathbf{x}_{-\gamma}(k)$ with the consideration of constraints (ones in \eqref{ap2} with different time index), and $(\mathbf{x}_1(k),\mathbf{x}_2(k))$ is an optimal strategy maximizing the connectivity under any $\psi$ link attacks. Therefore, the global MAS network is resistant to this type of strategic attacks during updates. We next proceed to show the second part. Through the designed algorithm, the network has been shown to converge to a meta-equilibrium. The network connectivity at the meta-equilibrium $(\mathbf{x}_1^*,\mathbf{x}_2^*)$ under attack $e^*\in BR(\mathbf{x}_1^*,\mathbf{x}_2^*)$ is $\lambda_2\big(\mathbf{L}^{e^*}_G(\mathbf{x}^*)\big)$. Since two designers have the same objective, at their equilibrium strategies $(\mathbf{x}_1^*,\mathbf{x}_2^*)$, both designers have a consistent anticipation on the set of link removals under the worst-case attack, i.e., $e^*$.   By the definition of Stackelberg game, we know that $\mathbf{x}_{\gamma}^*$ is the optimal strategy under any $\psi$ link removals for designer $\gamma$, and this fact holds simultaneously for both designers at equilibrium. Therefore, $\lambda_2\big(\mathbf{L}^{e^*}_G(\mathbf{x}^*)\big)$ is the maximum network connectivity that designers can achieve under strategic attacks, showing the resistance of the network to adversary at equilibrium. 
\end{proof}

We next investigate the uniqueness of the meta-equilibrium of the game, and the result is shown in the following theorem.
\begin{theorem}
The meta-equilibrium of the game is not unique, i.e., different equilibrium profiles $(\mathbf{x}_1^*,\mathbf{x}_2^*,e^*)$ are possible.
\end{theorem}
\begin{proof}
To show the nonuniqueness of the meta-equilibrium, one possible way is to find a different position pair $(\mathbf{\tilde{x}}_1^*,\mathbf{\tilde{x}}_2^*)$ but the network configuration is the same with a one under the meta-equilibrium, say $(\mathbf{x}_1^*,\mathbf{x}_2^*)$. This can be achieved by the simultaneous offset or rotation in $\mathbf{x}_1^*$ and $\mathbf{x}_2^*$. For example, under the meta-equilibrium $(\mathbf{x}_1^*,\mathbf{x}_2^*, e^*)$, the profile $(\mathbf{x}_1^*+\zeta,\mathbf{x}_2^*+\zeta, e^*)$ is also a meta-equilibrium, where $\zeta\in\mathbb{R}^3$, and this shows the nonuniqueness of the equilibrium.
\end{proof}
The network configuration at meta-equilibrium can also be different at which the network exhibits various levels of network connectivity. This phenomenon is further demonstrated through case studies in Section \ref{s5}.

\textit{Remark:} Due to nonuniqueness of meta-equilibrium, it is important for the network designers to achieve an equilibrium state with a higher connectivity. We need to deal with a mechanism design or an equilibrium selection problem. One general method to address this type of problem is through a top-down approach. Specifically, we first need to know the complete set of equilibria and then decentralize the solution for two network designers to drive the system to a desired equilibrium state. Hence, in addition to two network operators, we need to include an additional central planner who has perfect information of the system and guides the decentralized network configuration updates of two designers to achieve a particular equilibrium solution.

Another result is on the effectiveness of our proposed strategy comparing with simpler ones without attack anticipation for the network designers. In the proposed model considering adversary, if the attack does not occur, then the network connectivity achieved at the meta-equilibrium is no better than the one obtained by the model without considering adversaries. However, when the attack is successfully launched, the  network performance under the established framework is no worse than the one without considering adversaries. Characterizing the conditions under which these two classes of strategies coincide is not trivial and it is also related to the system parameters of multi-layer MAS. 

In Section \ref{s5}, we use a case study to show the advantage of our proposed framework over the traditional optimal design when the adversary is ignored. In the case study, the equilibrium network is still connected under the strategic attack while the optimal network counterpart is disconnected. Therefore, preparation for and reacting to attacks can yield significant benefits for the MAS network under adversarial environment.

\section{Adversarial Analysis}\label{s4}
In this section, we first analyze the security of MAS network by deriving a closed form solution of the jamming attacker, and then present another type of cyberattacks for further resiliency quantification of the proposed iterative algorithm. Finally, we discuss the benefits of anticipating and reacting to the adversary for the network designers during control design.

\subsection{Adversarial Analysis}\label{security_analysis}
Denote the network as $\widetilde{G}(i,j)=(V,E\setminus(i,j))$ after removing a link $(i,j)\in E$ from network $G$, then, we have $\widetilde{\mathbf{L}}=\mathbf{L}-\Delta \mathbf{L}$ and $\Delta \mathbf{L}=\Delta \mathbf{D}-\Delta \mathbf{A}$, where $\Delta \mathbf{D}$ and $\Delta \mathbf{A}$ are the decreased degree and adjacency matrices, respectively. By using equation \eqref{laplacian2}, we obtain $\Delta \mathbf{D}$ and $\Delta \mathbf{A}$ as follows:
$
\Delta \mathbf{D}=\mathbf{e}_i \tilde{\mathbf{e}}_{i,j}^{T}+\mathbf{e}_j \tilde{\mathbf{e}}_{j,i}^T,
\Delta \mathbf{A}=\mathbf{e}_i \tilde{\mathbf{e}}_{j,i}^T+\mathbf{e}_j \tilde{\mathbf{e}}_{i,j}^T,
$
where $\mathbf{e}_i$ and $\tilde{\mathbf{e}}_{i,j}$ are $n$-dimensional zero vectors except the $i$-th element equaling to 1 and $w_{ij}$, respectively, and similar for $\mathbf{e}_j$ and $\tilde{\mathbf{e}}_{j,i}$. Denote the Laplacian matrix of $\widetilde{G}(i,j)$ as $\widetilde{\mathbf{L}}(i,j)$, and by using $\Delta \mathbf{D}$ and $\Delta \mathbf{A}$, we have
\begin{equation}\label{linkattackL}
\widetilde{\mathbf{L}}(i,j)=\mathbf{L}-\big(\mathbf{e}_i-\mathbf{e}_j\big)\big(\tilde{\mathbf{e}}_{i,j}- \tilde{\mathbf{e}}_{j,i}\big)^T.
\end{equation}

When link $(i,j)$ is attacked, the resulting Laplacian is given by \eqref{linkattackL}. Denote the Fiedler vector of $\mathbf{L}$ as $\mathbf{u}$, and thus $\mathbf{u}^T\mathbf{Lu}=\lambda_2(\mathbf{L})$ based on the definition. By using Courant-Fisher Theorem in \eqref{lambdainf}, we  obtain the following:
\begin{equation}
\begin{split}
\lambda_2 \big(\widetilde{\mathbf{L}}(i,j)\big)&\le \mathbf{u}^T\widetilde{\mathbf{L}}(i,j) \mathbf{u}\\
&=\mathbf{u}^T \big( \mathbf{L}-\big(\mathbf{e}_i-\mathbf{e}_j\big)\big(\tilde{\mathbf{e}}_{i,j}-\tilde{\mathbf{e}}_{j,i}\big)^T \big) \mathbf{u}\\
&=\mathbf{u}^T\mathbf{Lu}-(u_i-u_j)(w_{ij}u_i-w_{ji}u_j)\\
&=\lambda_2(\mathbf{L})-w_{ij}(u_i-u_j)^2.
\end{split}
\end{equation}

Therefore, by removing the link $(i,j)^*$, where
\begin{equation}\label{worst_link_removal}
(i,j)^*\in\mathrm{arg} \max_{(i,j)\in{E}}\ w_{ij}(u_i-u_j)^2,
\end{equation}
the upper bound of $\lambda_2 \big(\widetilde{\mathbf{L}}(i,j)\big)$ is the smallest. The strategy in \eqref{worst_link_removal} can be seen as a greedy heuristic  for the attacker to compromise the network $G$. Specifically, the attacker can apply the above procedure iteratively to find a set of critical links to accommodate the attacker's ability. To this end, the jamming attacker's strategy is to compromise those links with top $\psi$ largest value of $w_{ij}(u_i-u_j)^2$, $i,j\in V$. Therefore, the network operators designs secure strategies by anticipating that these $\psi$ critical links could be compromised by the attacker.

\subsection{Another Type of Cyberattack}\label{two-other_a}
In order to assess the resilience of designed iterative algorithm in Section \ref{iterative_algorithm}, we introduce another type of adversarial attacks to the MAS network called global positioning system (GPS) spoofing attack.

\textit{GPS Spoofing Attack:} 
A GPS spoofing attack aims to deceive a GPS receiver in terms of the object's position, velocity and time by generating counterfeit GPS signals \cite{akos2012s}. In \cite{kerns2014unmanned}, the authors have demonstrated that UAVs can be controlled by the attackers and go to a wrong position through the GPS spoofing attack. We consider the scenario that the compromised robot is spoofed which can be realized by adding a disruptive position signal to the robot's real control command. Therefore, through the GPS spoofing attack, the mobile robot is controlled by the adversary, but it still maintains communications with other robots in the network. In addition, we assume that the attack cannot last forever but for a period of $g_a$ in the discrete time measure, since the resource of an attacker is limited, and the abnormal/unexpected behavior of the other unattacked robots resulting from the spoofing attack can be detected by the network operator.

Specifically, if robot $i$, $i\in V$, is compromised by the spoofing attack at time step $k_1$, and the attack lasts for $g_a$ time steps, then this scenario can be captured by adding the following constraint to $\overline{\mathcal{Q}}_\gamma^k$:
$
x_i(k+1)=x_i(k)+\epsilon(k),\ k=k_1,...,k_1+g_a-1,
$
where $\epsilon(k)$ is the disruptive signal added by the attacker.
The attacked robot is usually randomly chosen. To evaluate the impact of attack, we choose the robot that has the maximum degree  denoted by ${i}_{max}$ and satisfies
\begin{align}
{i}_{max}\in\ \mathrm{arg}\ \max_{i\in V}\sum_{j\in \mathcal{N}_i}w_{ij},
\end{align}
where $\mathcal{N}_i$ is the set of nodes connected to robot $i$. 

The GPS spoofing attack decreases the network connectivity. The resilience of the designed algorithm can be quantified by the increased network performance by the network operators' responses to the cyberattacks.

\subsection{Benefits of Anticipating and Reacting to Adversary}
We comment on the benefits of the established approach that enhances the network resistance to attacks.
First, the Stackelberg modeling incorporates the security consideration of network designers. This anticipative behavior of the defender ensures the resistance of network under successful worst-case attacks. The security consideration is critical and also practical for the network operators in devising proactive defense strategies, and its significance is demonstrated through case studies in Section \ref{s5}. Second, besides the inherent security modeling, reacting to the adversary also helps in improving the system performance if the attack is unanticipated. Then, the operator can respond to the unanticipated attack in the next round which enhances network resilience. Thus, the proposed method guides the secure and resilient decentralized control design of MAS networks.

\section{Case Studies}\label{s5}
In this section, we use case studies to quantify the security and resiliency of the designed algorithm, and identify the interdependency in the multi-layer MAS networks. We adopt YALMIP \cite{lofberg2004yalmip} to solve the corresponding SDP problems. Specifically, we consider a two-layer MAS network in which $G_1$ contains 2 nodes and $G_2$ contains 6 nodes.
To illustrate that the designed framework can be applied to cases where the robots at one layer can further be operated in a decentralized way, we assume that the robots in $G_2$ are divided into 2 equal-size groups connected by a secure link between nodes 3 and 4. The investigated scenario is applicable when the agents in MAS are sparsely distributed in geometric clusters. 

The communication strength between agents follows the one in Fig. \ref{Aij}. Further, two layers of MAS are operated in two planes where the third dimension of their position is fixed satisfying the minimum distance $\rho_{12}$. The initial positions of robots in $G_1$ (upper layer) are (1,3,1.2), (2,3,1.2), and robots in $G_2$ (lower layer) are (0,0,0), (0,1.5,0), (1,0,0), (2,0,0), (3,-1.5,0), (3,0,0). The safety distance between robots in $G_1$ and $G_2$ is $\rho_1=\rho_2=1$, and the maximum distance that robots at each layer can move at each update step is $d_1=d_2=0.2$. The update frequency of network operator $P_1$ is two times faster than network operator $P_2$, i.e., $2c_1=c_2$. In addition, both network designers prepare for the worst-case single link removal of jamming attack, i.e., $|e|=\psi=1$, during the MAS network formation.

\subsection{Secure Design of MAS Networks}\label{secure_config}
First, we illustrate the secure design of two-layer MAS network under the jamming attack using Algorithm \ref{algorithm3}. Fig. \ref{base_case} shows the results. The final positions of agents in $G_1$ are (1.88,-0.13,1.2), (1.88,0.87,1.2), and those in $G_2$ are (0.94,-0.85,0), (0.94,0.65,0), (1.60,-0.10,0), (2.26,0.65, 0), (1.55,-0.29,0), (2.92,1.40,0). The connectivity of the integrated MAS network is iteratively improved, and converges to a steady value 1.4 after approximate 40 steps. During the updates, each network operator anticipates of the strategic jamming attack that compromises the most critical communication link. Hence, the network shown in Fig. \ref{base_case_final_config} is a meta-equilibrium configuration. This example shows the effectiveness of the proposed method in designing secure multi-layer MAS network. To show the nonuniqueness of the equilibrium solutions, we modify the initial positions of agents where the robots in $G_1$ (upper layer) start with (3,1,1.2) and (3,2,1.2). The results are shown in Fig. \ref{base_case_2_for_compare}. We can see that the final positions of agents in $G_1$ are (1.06,0.59,1.2), (2.06,0.59,1.2), and those in $G_2$ are (0.05,0.58,0), (1.55,1.45,0), (1.05,0.58,0), (2.05,0.58, 0), (1.55,-0.29,0), (3.05,0.58,0). Further, the final network configuration as well as network connectivity are different with the ones shown in Fig. \ref{base_case} which corroborate the meta-equilibrium of the proposed game is not unique.

Another critical question is to compare the quality of the solution at meta-equilibrium with the one obtained without considering attacks. To better illustrate the results, we select another set of parameters for communication channels: $\delta=0.1$, $c_1 = 1$, and $c_2=1.2$, where agents have a lower range of communication than the one in Fig. \ref{base_case}. In addition, the number of agents in the upper layer $G_1$ is reduced to 1 and the number of attacks is equal to $\psi = 2$. The adversary's action set $\mathcal{A}$ includes the links between two layers, i.e., inter-links. Other settings are the same as those in previous case studies. The results of network configurations with and without considering attacks are shown in Fig.  \ref{base_case_4_for_compare}. In Fig. \ref{base_case_trajectory_4_for_compare}, since no adversary is present in this scenario, the network connectivity  remains around 0.48. To demonstrate the difference of the optimal solution without considering attack and the equilibrium solution, we introduce two attacks by removing two most critical inter-links using the strategy in Section \ref{security_analysis} to the optimal network (removing links (3,7) and (4,7)) at step 50, and the result is shown in Fig. \ref{base_case_connectivity_4_for_compare}. Comparing with the secure MAS configuration using the proposed control with connectivity 0.17 shown in Fig. \ref{base_case_final_config_4_for_compare_secure_one}, the originally optimal network is completely disconnected after the attack which makes the isolated agent 7 in an unanticipated condition.  This result shows the advantage of our proposed framework over the traditional optimal design, since our approach integrates the security aspects in the MAS control design.

\begin{figure}[t]
  \centering
    \subfigure[]{
    \includegraphics[width=0.6\columnwidth]{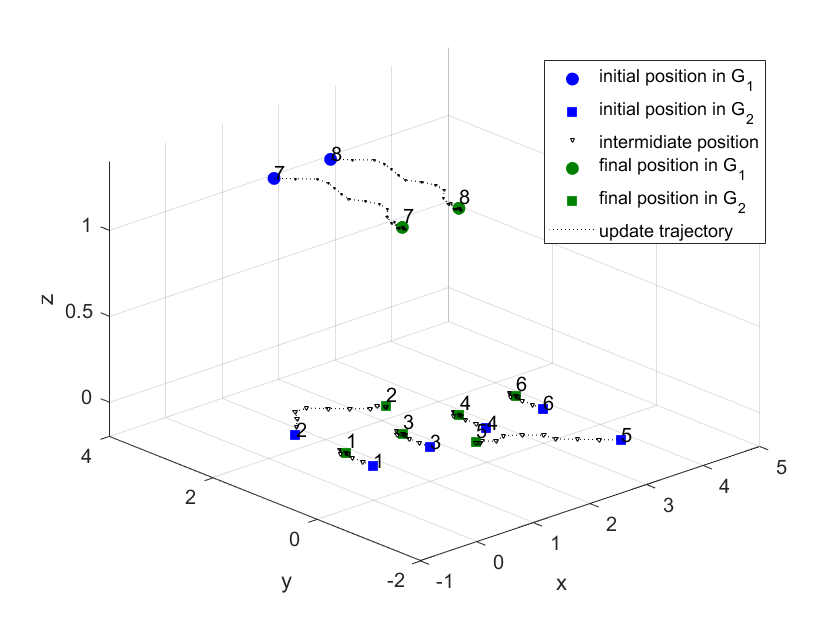}} 
  \subfigure[]{
    \includegraphics[width=0.47\columnwidth]{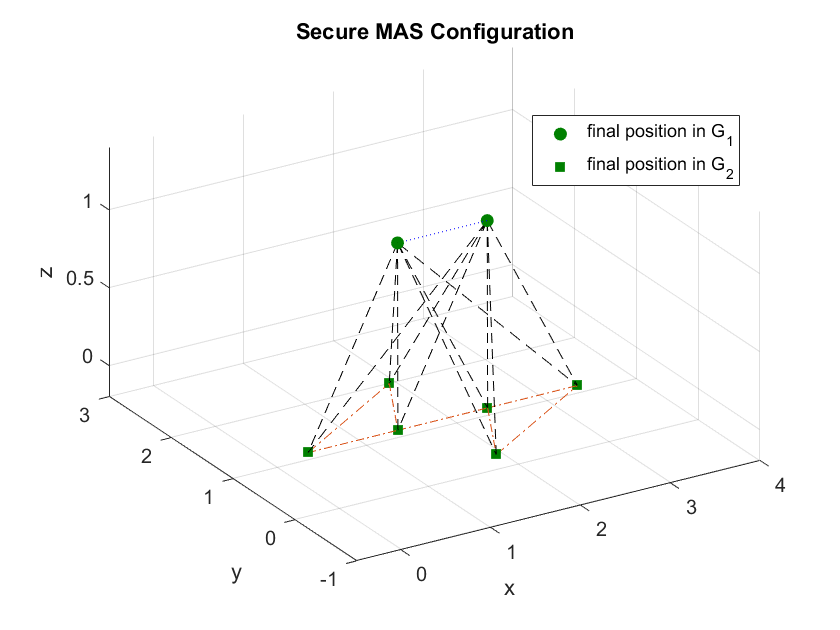}  \label{base_case_final_config}} 
      \subfigure[]{
    \includegraphics[width=0.47\columnwidth]{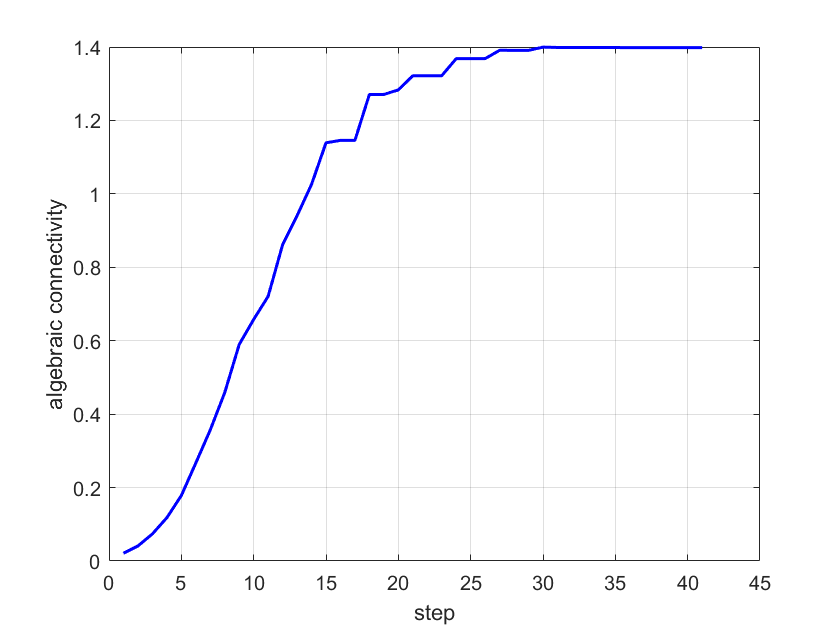}} 
  \caption[]{(a) shows the evolutionary configuration of secure MAS network at each step. (b) depicts the final network configuration. (c) shows the network connectivity under attack with $\psi = 1$.}
  \label{base_case}
\end{figure}

\begin{figure}[t]
  \centering
    \subfigure[]{
    \includegraphics[width=0.6\columnwidth]{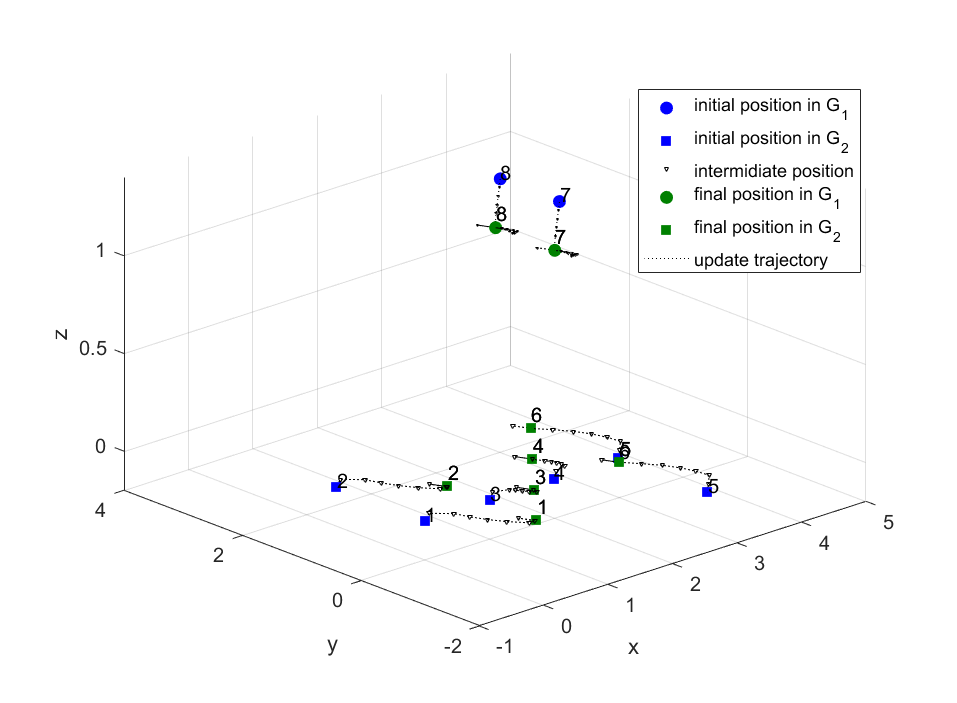}} 
  \subfigure[]{
    \includegraphics[width=0.47\columnwidth]{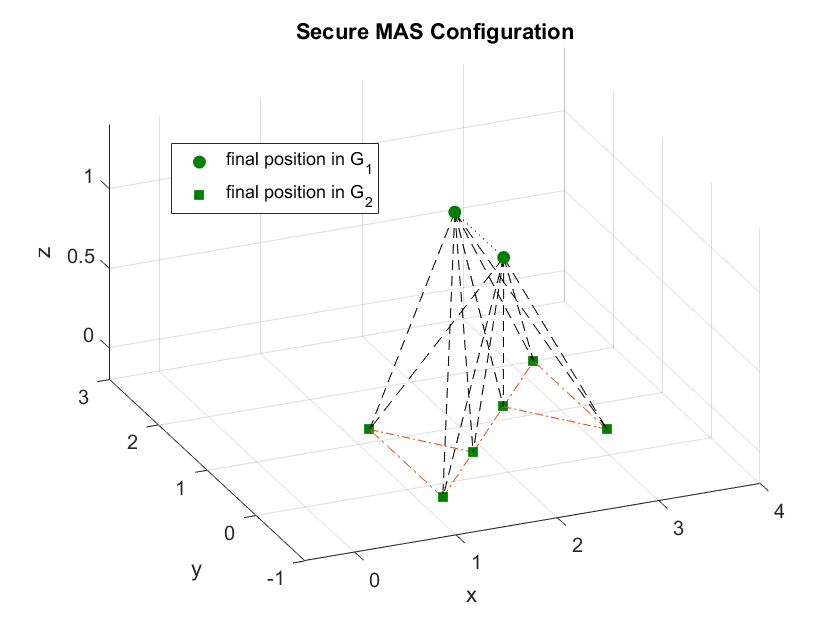}  \label{base_case_final_config_2_for_compare}} 
      \subfigure[]{
    \includegraphics[width=0.47\columnwidth]{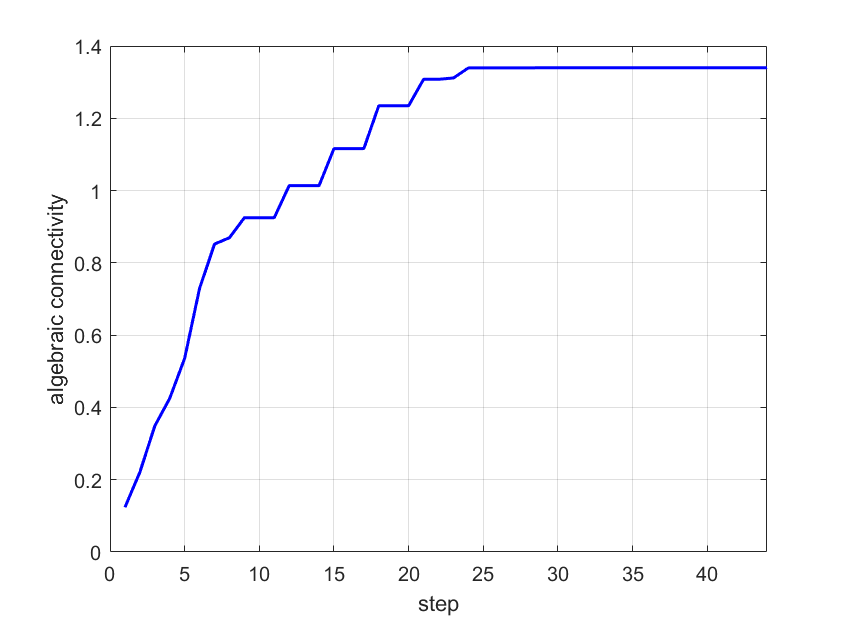}} 
  \caption[]{(a), (b), and (c) show the results of the ones in Fig. \ref{base_case}. The initial conditions of agents are modified and the final equilibrium network is different with the one in Fig. \ref{base_case} which shows the nonuniqueness of the equilibrium.}
  \label{base_case_2_for_compare}
\end{figure}

\begin{figure}[t]
  \centering
    \subfigure[]{
    \includegraphics[width=0.47\columnwidth]{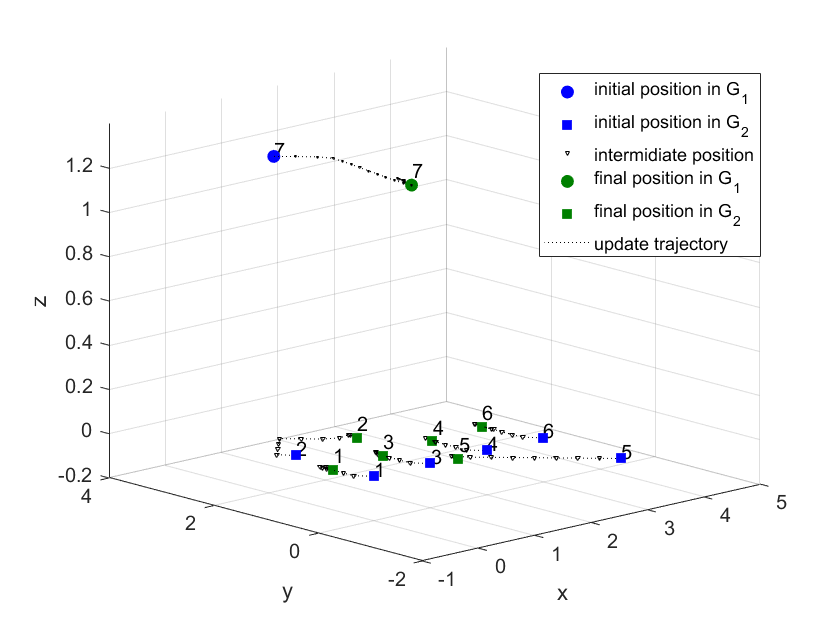}\label{base_case_trajectory_4_for_compare}} 
      \subfigure[]{
    \includegraphics[width=0.47\columnwidth]{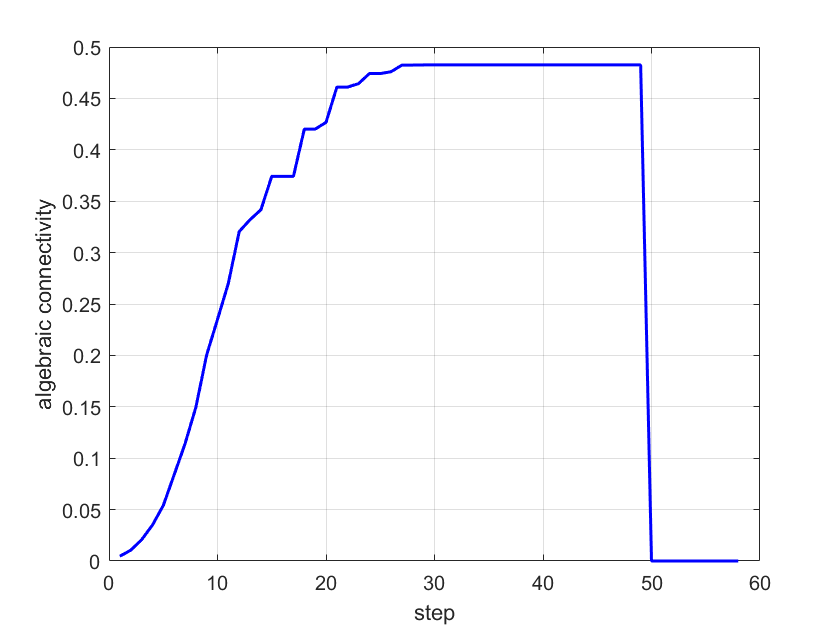}\label{base_case_connectivity_4_for_compare}} 
     \subfigure[]{
    \includegraphics[width=0.46\columnwidth]{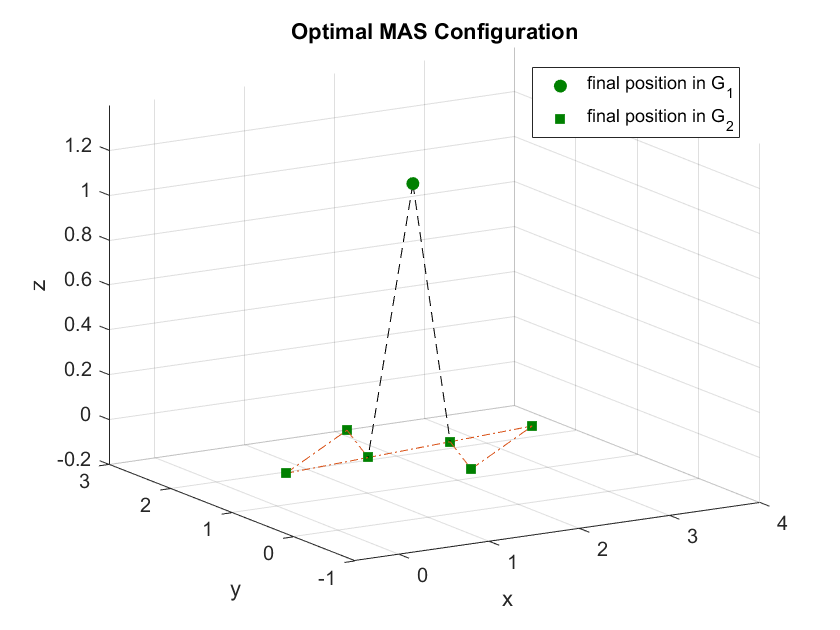}
    \label{base_case_final_config_4_for_compare}} 
         \subfigure[]{
    \includegraphics[width=0.47\columnwidth]{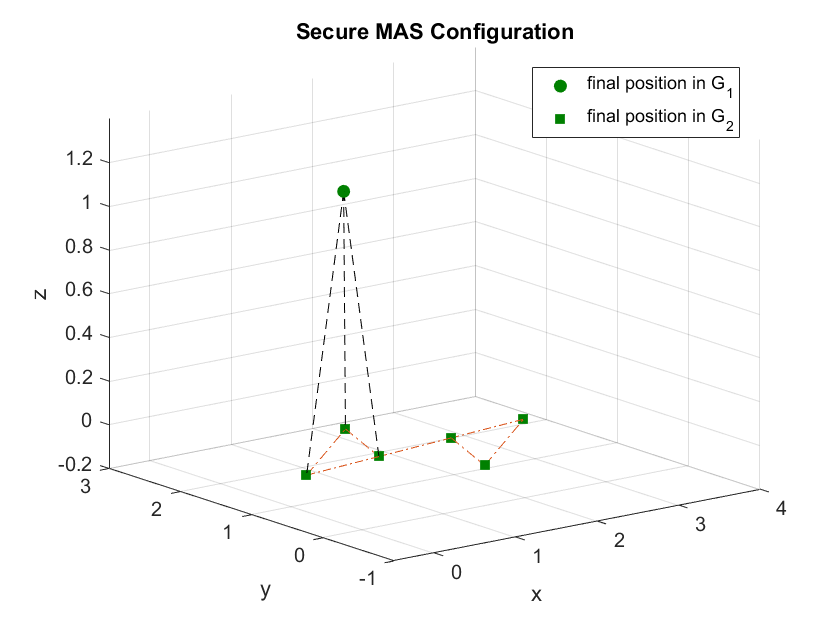}
    \label{base_case_final_config_4_for_compare_secure_one}}
  \caption[]{(a), (b), and (c) show the results when the network designers do not anticipate attacks. (d) is the equilibrium network when designers considers two link attacks. The final network configurations in (c) and (d) are different. After introducing the worst-case attacks to the optimal network (compromising the inter-links (3,7) and (4,7)), the optimal network in (c) is disconnected and the connectivity becomes 0. However, the meta-equilibrium network in (d) is still connected after 2 inter-links attacks with a connectivity of 0.17 which demonstrates the enhanced security of MAS networks. }
  \label{base_case_4_for_compare}\vspace{-3mm}
\end{figure}

\subsection{Resilience of the Network to Cyberattacks}
Second, we investigate the resilience of the designed MAS network to cyberattacks presented in Section \ref{two-other_a}. The metric used for quantifying the resilience is the recovery ability of network connectivity after the adversarial attack.

For the GPS spoofing attack, we assume that it lasts for 5 time steps from step 9 to 14 before the detection of abnormal movement of MAS by the network operator. Note that the attack duration depends on the detection ability of the network designer. After the identification of attack, the network designer can reboot the compromised agent for it returning to the normal state.  Moreover, the horizontal axis in attacker's disruptive command $\epsilon(k)$, $k=9,...,14$, is drawn uniformly from $[0,0.2]$. Fig. \ref{spoofing_case} shows the obtained results where agent 7 is compromised. Specifically, the network connectivity encounters a sudden drop at step 9, from 0.58 to 0.37, as shown in Fig. \ref{spoofing_connectivity} due to the spoofing attack. At step 12 which is still in the attacking window, the connectivity, however, has an increase which is a result from the updates of agents at the lower layer $G_2$. When the spoofing attack is removed, the network recovers quickly after step 14 which shows agile resilience of the proposed control algorithm. Note that the final MAS network configuration is the same as the one in Section \ref{secure_config}. 

We next investigate a scenario in which the spoofing attack is introduced after the network reaching an equilibrium. Specifically, the attack launches at step 35 and it lasts for 6 steps. The results are shown in Fig. \ref{spoofing_case_equilibrium_attack}. Similar to the previous case, the network can response to the attack in a fast fashion and tries to recover the network connectivity with its best effort during the attacking window. After the detection and removal of the attack, the network performance is improved and the equilibrium network configuration is achieved which is the same as the previous one before attack. Thus, the designed two-layer MAS network using Algorithm \ref{algorithm3} is resilient to spoofing cyberattacks.

\begin{figure}[t]
  \centering
  \subfigure[]{
    \includegraphics[width=0.48\columnwidth]{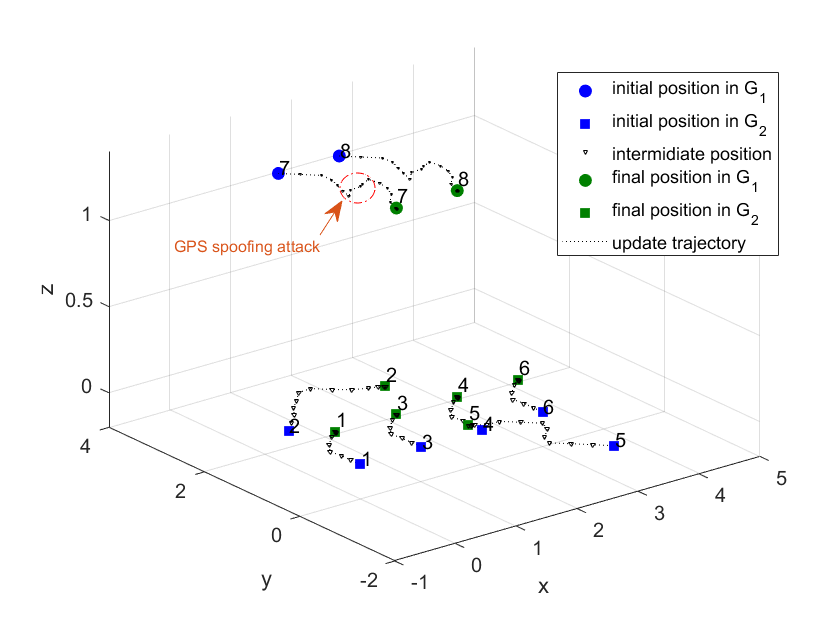}} 
      \subfigure[]{
    \includegraphics[width=0.48\columnwidth] {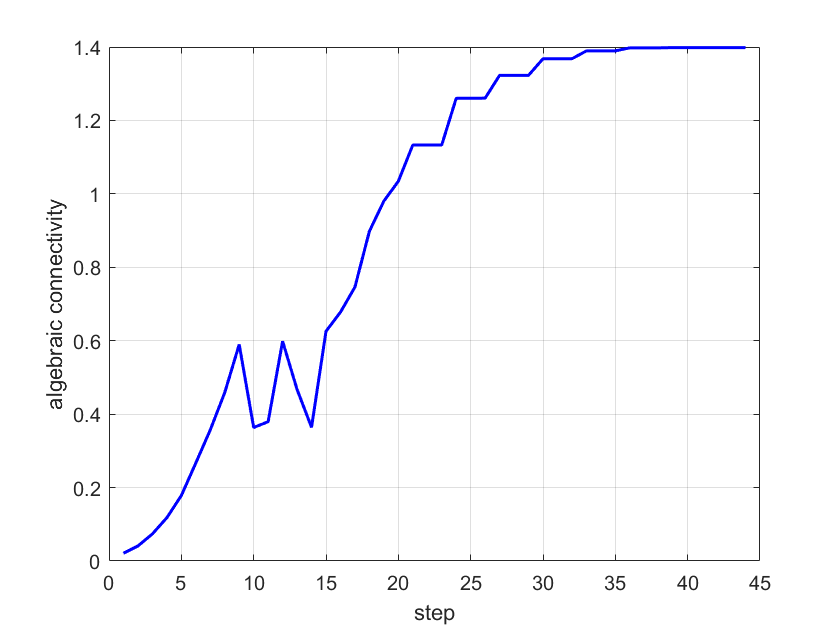}\label{spoofing_connectivity}}
  \caption[]{(a) shows the evolutionary configuration of secure MAS network at each step. The GPS spoofing attack is introduced at time step 9, and it lasts for 5 steps. The attack duration depends on the detection ability of the network designer.  (b) shows the corresponding network connectivity.}
  \label{spoofing_case}
\end{figure}

\begin{figure}[t]
  \centering
  \subfigure[]{
    \includegraphics[width=0.48\columnwidth]{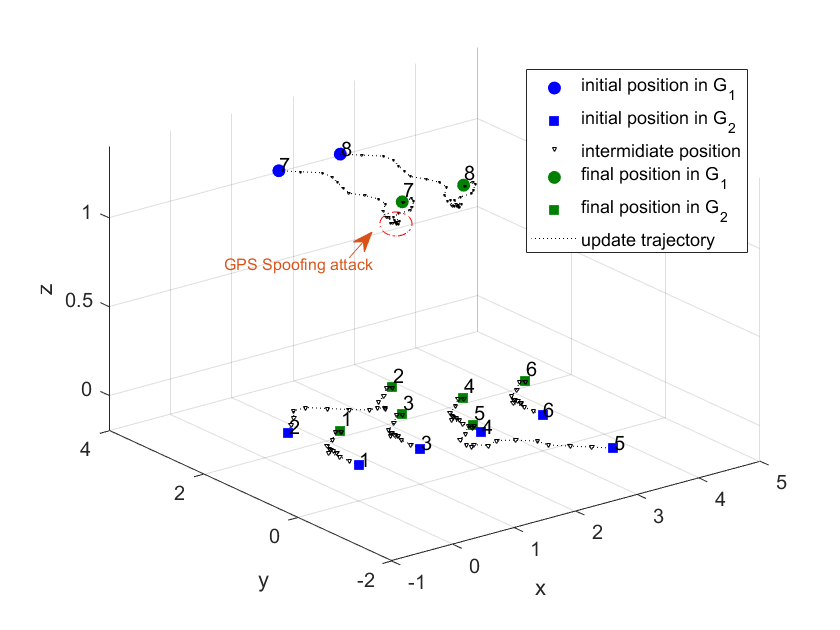}} 
      \subfigure[]{
    \includegraphics[width=0.48\columnwidth]{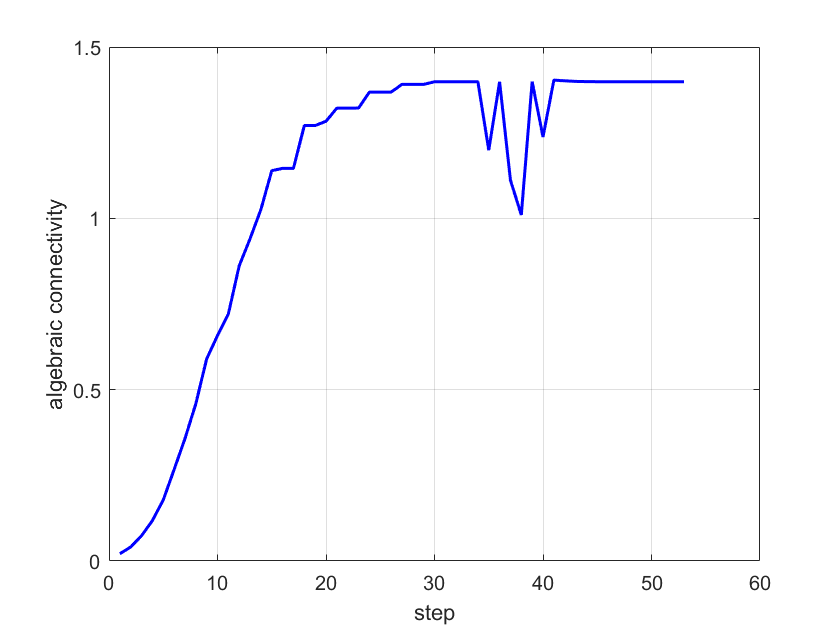}\label{spoofing_connectivity_equilibrium_attack}}
  \caption[]{(a) shows the evolutionary configuration of secure MAS network at each step. (b) shows the corresponding network connectivity. The spoofing attack launches at step 35 and it lasts for 6 steps. The network recovers and reaches a meta-equilibrium quickly after the removal of attack.}
  \label{spoofing_case_equilibrium_attack}\vspace{-3mm}
\end{figure}

\section{Conclusion}\label{conclusion}
In this paper, we have investigated the secure control of multi-layer MAS networks under the adversarial environment by establishing a games-in-games framework. The newly proposed meta-equilibrium solution concept has successfully captured the secure and uncoordinated design of each layer of MAS network through integrative Stackelberg and Nash games. The developed iterative algorithm has been shown effective in maximizing the network algebraic connectivity under adversaries, yielding a meta-equilibrium network configuration. Case studies have shown that the designed multi-layer MAS network is of agile resilience to various kinds of cyberattacks. As for future work, we can consider the network operators having different estimations of severity of attacks, and design the multi-layer MAS network with heterogeneous security requirements. Another research direction is to design mechanisms and decentralized algorithms to drive the multi-layer MAS to a desired meta-equilibrium if multiple equilibria exist. We can also investigate other games-in-games frameworks where the attacker's behavior is more sophisticated, e.g., the attacker makes decisions over a time horizon.

\appendices

\section{Derivation of Equation \eqref{discretedistance}}\label{appDerive}
The dynamics of each robot $i$ is originally described by continuous-time models. Denote $\mathcal{Z}_{ij}(t) = ||x_{ij}(t)||_2^2$
and then differentiate $||x_{ij}(t)||_2^2$ with respect to time $t$, and we obtain
\begin{equation}\label{differentiate}
2\{\dot{x}_i(t)-\dot{x}_j(t)\}^T\{x_i(t)-x_j(t)\}=\dot{\mathcal{Z}}_{ij}(t).
\end{equation}
By using Euler's first order method, i.e.,
$
x(t)\rightarrow x(k),\ \dot{x}(t)\rightarrow \frac{x(k+m)-x(k)}{m \tau},
$
where $m\in\mathbb{N}^+$, and $m \tau$ is the time interval between two updates, we have
$ 2\{\dot{x}_i(t)-\dot{x}_j(t)\}^T\{x_i(t)-x_j(t)\}\cdot m \tau
=2\{{x}_i(k+m)-{x}_j(k+m)\}^T\{x_i(k)-x_j(k)\} -2||x_i(k)-x_j(k)||_2^2
=2\{{x}_i(k+m)-{x}_j(k+m)\}^T\{x_i(k)-x_j(k)\}-2{\mathcal{Z}}_{ij}(k),$
and 
${\mathcal{\dot Z}}_{ij}(t)\cdot m \tau={\mathcal{Z}}_{ij}(k+m)-{\mathcal{Z}}_{ij}(k).$
Thus, we can rewrite \eqref{differentiate} as
$
{\mathcal{Z}}_{ij}(k+m)+{\mathcal{Z}}_{ij}(k) = ||x_{ij}(k+m)||_2^2 + ||x_{ij}(k)||_2^2 = 2\{{x}_i(k+m)-{x}_j(k+m)\}^T\{x_i(k)-x_j(k)\}
$.

\bibliographystyle{IEEEtran}
\bibliography{IEEEabrv,references}

\end{document}